\documentclass[reqno]{amsart}


\usepackage{amsmath}
\usepackage{amssymb}
\usepackage{amsthm}
\usepackage{amsfonts}
\usepackage{enumerate}
\usepackage{mathtools}
\usepackage{bm}
\usepackage{ulem}
\usepackage{fancyhdr}
\usepackage{hyperref}
\usepackage{mdframed}
\usepackage{tcolorbox}
\usepackage{graphicx,lipsum,wrapfig}
\usepackage[font={tiny}]{caption}
\usepackage[font={scriptsize}]{caption}

\usepackage{xspace}
\usepackage{subcaption}
\usepackage{todonotes}
\usepackage{mathtools}
\usepackage{xcolor}
\usepackage{wrapfig}
\usepackage{lipsum}

\usepackage{algorithm}
\usepackage{algpseudocode}
\algtext*{EndFor}
\algtext*{EndIf}
\usepackage{blindtext}
\usepackage{multicol}

\theoremstyle{plain}
\newtheorem{theorem}{Theorem}[section]
\newtheorem{corollary}[theorem]{Corollary}
\newtheorem{lemma}[theorem]{Lemma}
\newtheorem{proposition}[theorem]{Proposition}
\newtheorem*{theorem*}{Theorem}
\newtheorem{remark}[theorem]{Remark}
\newtheorem{definition}[theorem]{Definition}
\newtheorem{problem}[theorem]{Problem}
\newtheorem{assumption}[theorem]{Assumption}

\numberwithin{equation}{section}

\renewcommand{\emph}[1]{\textit{#1}}

\usepackage{times}
\usepackage{booktabs}
\usepackage{dsfont}
\usepackage[T1]{fontenc}
\usepackage[utf8]{inputenc}
\usepackage{eso-pic}

\newcommand{\ig}[1]{\textcolor{blue}{[IG: #1]}}
\newcommand{\ml}[1]{\textcolor{orange}{[ML: #1]}}
\newcommand{\mm}[1]{\textcolor{red}{[MM: #1]}}
\newcommand{\aabate}[1]{\textcolor{magenta}{[AA: #1]}}
\newcommand{\luca}[1]{\textcolor{magenta}{[LL: #1]}}
\renewcommand{\path}{\omega}
\newcommand{\x}{\boldsymbol{{x}}} 
\newcommand{\w}{\boldsymbol{{w}}}
\newcommand{\pathX}{\path_{x}}
\newcommand{\pathXbold}{\boldsymbol{\path}_{x}}
\newcommand{\PathX}{\Omega_{x}}
\newcommand{\pu}{\textbf{u}}
\newcommand{\policy}{\sigma}
\newcommand{\Setofpolicies}{\Sigma}
\newcommand{\adversary}{\xi}
\newcommand{\Adversary}{\Xi}
\newcommand{\PathFin}{\Omega_{x}^{\rm{fin}}}
\newcommand{\pq}{\textbf{q}}
\newcommand{\pathrmdp}{\omega}
\newcommand{\pathrmdpbold}{\boldsymbol{\omega}}
\newcommand{\Pathrmdp}{\Omega}
\newcommand{\PRMDP}{\M^\varphi}
\newcommand{\pathprmdp}{\path_\PRMDPS}
\newcommand{\pathrmdpfin}{\pathrmdp^{\mathrm{fin}}}
\newcommand{\Pathrmdpfin}{\Omega^{\mathrm{fin}}}
\newcommand{\last}{\mathrm{last}}
\newcommand{\reals}{\mathbb{R}}
\newcommand{\naturals}{\mathbb{N}}
\newcommand{\M}{\mathcal{M}}
\newcommand{\I}{\mathcal{I}}
\newcommand{\Inom}{\widehat{\mathcal{I}}}
\newcommand{\D}{\mathcal{D}}
\newcommand{\A}{\mathcal{A}}
\newcommand{\LTLf}{LTL$_f$\xspace}
\newcommand{\LTL}{LTL\xspace}
\newcommand{\Globally}{\Box}
\newcommand{\Eventually}{\Diamond}
\newcommand{\luntil}{\mathcal{U}}
\newcommand{\Next}{\bigcirc}
\newcommand{\trace}{\theta}
\newcommand{\APs}{AP}
\newcommand{\Prop}{\mathfrak{p}}
\newcommand{\safe}{\mathrm{safe}}

\begin{document}

\AddToShipoutPictureBG*{%
  \AtPageUpperLeft{%
    \hspace{16.5cm}%
    \raisebox{-1.5cm}{%
      \makebox[0pt][r]{To appear in \textit{Learning for Dynamics and Control Conference} (\textit{L4DC}), June 2025.}}}}

\title[\LTLf Control under Unknown Disturbances]{
Temporal Logic Control for Nonlinear Stochastic Systems Under Unknown Disturbances
}
\author[I. Gracia]{Ibon Gracia}

\author[L. Laurenti]{Luca Laurenti}

\author[M. Mazo]{Manuel Mazo Jr.}

\author[A. Abate]{Alessandro Abate}

\author[M. Lahijanian]{Morteza Lahijanian}

\thanks{The authors are  with the Smead
Department of Aerospace Engineering Sciences, CU Boulder, Boulder, CO, \{ibon.gracia, morteza.lahijanian\}@colorado.edu, 
the Delft Center for Systems and Control, 
Faculty of Mechanical, Maritime and Materials Engineering, Delft University of Technology \{m.mazo,l.laurenti\}@tudelft.nl and the Department of Computer Science, University of Oxford alessandro.abate@cs.ox.ac.uk}

\begin{abstract}%
    In this paper, we present a novel framework to synthesize robust 
    strategies for discrete-time nonlinear 
    systems with random disturbances that are unknown, against temporal logic specifications. 
    The proposed framework is data-driven and abstraction-based: leveraging observations of the system, our approach learns a high-confidence abstraction of the system in the form of an uncertain Markov decision process (UMDP). The uncertainty in the resulting UMDP is used to formally account for both the error in abstracting the system and for the uncertainty coming from the data. 
    Critically, we show that for any given state-action pair in the resulting UMDP, the uncertainty in the transition probabilities can be represented as a convex polytope obtained by a two-layer state discretization and concentration inequalities. This allows us to obtain tighter uncertainty estimates compared to existing approaches, and guarantees efficiency, as we
    tailor a synthesis algorithm exploiting the structure of this UMDP.
    We empirically validate our approach on several case studies, showing substantially improved performance compared to the state-of-the-art.  
    
\end{abstract}

\maketitle


\section{Introduction}
\label{sec:intro}

The synthesis of safe strategies for stochastic systems is critical in ensuring \emph{reliable} and \emph{safe} operations in domains such as robotics, autonomous vehicles, and cyber-physical systems \cite{belta2005discrete,lavaei2022automated}. A key challenge arises when the system dynamics include \emph{unknown} random disturbances, making it difficult to account for uncertainties while guaranteeing performance against high-level complex specifications. Existing methods often assume known distributions for the disturbances or rely on abstractions with overly conservative uncertainty estimates, limiting their scalability and applicability to complex systems. This paper aims to address these gaps by presenting a novel framework to synthesize optimal strategies for nonlinear stochastic systems with unknown disturbances, ensuring both formal guarantees and computational efficiency.  


Our framework employs a data-driven, abstraction-based approach to strategy synthesis for stochastic systems with unknown noise under \textit{linear temporal logic over finite traces} (\LTLf) \cite{de2013linear} specifications. Starting with data from the system's trajectories, we construct a high-confidence abstraction in the form of an uncertain Markov decision process (UMDP) \cite{iyengar2005robust}, a flexible model that captures complex uncertainties. Unlike existing methods relying on interval-based abstractions or conservative assumptions, our framework represents transition probability uncertainties as convex polytopes. These sets are derived through a novel two-layer discretization scheme and learning the support of the unknown disturbance. This leads to tighter uncertainty sets and less conservative results compared to existing methods. Exploiting this UMDP structure, we introduce a synthesis algorithm for \LTLf specifications that simplifies the computation, reducing the complexity of standard UMDP linear programming approaches. By incorporating uncertainties from both abstraction errors and data limitations, our framework yields a strategy that is robust. Our empirical evaluations over various types of systems reveals the efficacy of this approach over existing methods, namely in data efficiency, tightness of results, and scalability.

The main contributions of this paper are fourfold: (i) a novel framework for synthesizing strategies for nonlinear stochastic systems under non-additive, unknown disturbances with \LTLf specifications, (ii) a distribution-agnostic, data-driven construction of UMDP abstraction with a specific structure that reduces conservatism of existing abstraction-based techniques, (iii) an efficient tailored synthesis algorithm for this UMDP abstraction, which does not introduce additional conservatism
, and (iv) a series of case studies and benchmarks that show superiority of the framework over the state-of-the-art, with up to $3$ orders of magnitude improvement in sample complexity and an order of magnitude reduction in computation time.

\paragraph{Related Work}

Abstractions of stochastic systems to finite Markov decision processes (MDPs) are powerful tools for controller synthesis on highly-complex systems under complex logic specifications \cite{lavaei2022automated}. In particular, 
Interval MDPs (IMDPs) \cite{lahijanian2015formal,givan2000bounded} abstract systems by presenting uncertain transition probabilities within intervals, capturing the full range of system behaviors. For example, \cite{cauchi2019efficiency} efficiently abstracts linear systems with additive Gaussian noise, while \cite{skovbekk2023formal} extends this to nonlinear dynamics. Uncertain MDPs (UMDPs) \cite{iyengar2005robust,el2005robust} generalize IMDPs by allowing transition probabilities to belong to more complex sets and have been used for strategy synthesis against specifications such as linear temporal logic (LTL) \cite{wolff2012robust}. However, these abstractions typically require system models, which are often unavailable in practice.

To address model uncertainty, various methods leverage Gaussian processes \cite{jackson2021strategy}, neural networks \cite{adams2022formal}, and ambiguity sets \cite{gracia2022distributionally}, which are then abstracted as IMDPs or UMDPs. Statistical tools like the scenario approach have also been used to abstract stochastic \cite{badings2023robust,badings2023probabilities}, non-deterministic \cite{kazemi2024data}, and deterministic systems \cite{coppola2022data,coppola2023data}. Also, techniques such as super-martingales and barrier functions enable safety verification and control synthesis for general dynamics \cite{lechner2022stability,badings2024learning,mazouz2024piecewise}. 
Nevertheless, all these works assume that the disturbance distribution is known.

When disturbance distributions are uncertain, some works combine IMDP abstractions with statistical tools 
\cite{badings2023robust,badings2023probabilities,schon2023bayesian}, while others employ barrier certificates with the scenario approach for safety verification \cite{mathiesen2023inner}. Another approach  \cite{gracia2025efficient} constructs Wasserstein ambiguity sets from data samples to abstract systems as UMDPs, which account for the uncertainty regarding the unknown distributions. However, these methods typically assume simple dynamics or additive noise. 
For general dynamics with unknown disturbance distributions, only a few works exist. \cite{salamati2021data} uses barrier certificates for safety verification, and \cite{gracia2024data} extends the ambiguity set approach to nonlinear dynamics, formally characterizing the guarantees of these approaches. Both assume that certain distribution-related properties, such as variance or support, are known—an assumption often unrealistic in practice, and also suffer from high sample complexity, especially under high-confidence requirements.  
Our work overcomes these limitations by removing assumptions about disturbance distributions and offering a data-efficient and scalable approach suitable for systems with general dynamics.
\section{Problem Formulation}

In this work the focus is on discrete-time stochastic systems given by
\vspace{-2mm}
\begin{align}
\label{eq:sys}
    \x_{k+1} = f(\x_k,u_k,\w_k),
\end{align}
where $\x_k\in\reals^n$ denotes the state at time $k\in\naturals_0$, $u_k\in U  \subset \reals^m$ is the control input chosen from a finite set $U$, and $\w_k\in W \subseteq \reals^d$ is the disturbance. The latter is a sequence of independent and identically distributed (i.i.d.) random variables on the probability space $(W,\mathcal{B}(W), P_W)$, with $\mathcal{B}$ being the Borel $\sigma$-algebra on $W$, and where the support $W$ and probability distribution $P_W$ of $\w_k$ are \emph{unknown}. The vector field (possibly nonlinear) $f:\reals^n\times U\times W \rightarrow \reals^n$ is assumed to be Lipschitz continuous on its third argument, uniformly for all values of its first argument on some set. 

\vspace{-2mm}
\begin{assumption}
\label{ass:lipschitz}
    There exists a set $X\subset\reals^n$, such that, for every $u\in U$, there exists a constant $L_u >0$ 
    such that, for all $x\in X$, $w,w'\in W$, it holds that $\|f(x, u, w) - f(x, u, w')\| \le L_u \|w - w'\|$.
\end{assumption}
\vspace{-2mm}
In lieu of unknown $W$ and $P_W$, we assume a dataset on the disturbance is available.
\vspace{-2mm}
\begin{assumption}
\label{ass:samples}
    A set $\{\boldsymbol{\hat w}^{(i)}\}_{i=1}^N$ of $N$ i.i.d. samples from $P_W$ is available.
\end{assumption}
%
Assumption~\ref{ass:samples} is commonly made in related work \cite{badings2023robust,gracia2024data} and can be practically satisfied through, e.g., observations of the state and control. The straightforward example is when $f$ is affine in $\w_k$; otherwise, it suffices for $f$ to be injective over only a subset of $\mathbb{R}^n$ as discussed in \cite{gracia2024data}. This condition is met by many practical systems, including those in our case studies.


Given $x_0,\ldots,x_K \in \mathbb{R}^n$, $u_0,\dots,u_{K-1}\in U$, and $K \ge 0$, we denote a finite \emph{trajectory} of System~\eqref{eq:sys} by $\pathX =x_0 \xrightarrow{u_0}
\ldots \xrightarrow{u_{K-1}} x_K$. We let $|\pathX|$ denote the length of $\pathX$, define $\PathX$ as the set of all trajectories with $|\pathX|<\infty$ and denote by $\pathX(k)$ 
the state of $\pathX$ at time $k\in\{0,\dots, |\pathX|\}$. 
 A \emph{strategy} of System~\eqref{eq:sys} is a function $\policy_x: \PathX \to U$ that assigns a control $u$ to each finite trajectory $\pathX$. 
Given $x\in\reals^n$, $u\in U$, the \emph{transition kernel} $T:\mathcal{B}(\reals^n)\times \reals^n\times U \rightarrow [0, 1]$
of System~\eqref{eq:sys} assigns the probability
$T(B\mid x,u) = \int_W \mathds{1}_{B}(f(x,u,w))P_W(dw)$, where the indicator function $\mathds{1}_B(x) = 1$ if $x \in B$, and $0$ otherwise,
to each Borel set $B\in\mathcal{B}(\reals^n)$. For a strategy $\policy_x$ and an initial condition $x_0\in\reals^n$, the transition kernel defines a unique probability measure $P_{x_0}^{\policy_x}$ over the trajectories of System~\eqref{eq:sys} \cite{bertsekas1996stochastic}.
In this way, $P_{x_0}^{\policy_x}[\pathXbold(k) \in B]$ denotes the probability that $\x_k$ belongs to the set $B\in \mathcal{B}(\reals^n)$ when following strategy $\policy_x$ from initial state $x_0$. 
In this work, we are interested in the temporal behavior of System~\eqref{eq:sys} w.r.t. a bounded (\emph{safe}) set $X\in \mathcal{B}(\reals^n)$ and a set of regions of interest $R_\text{int}
$, with $r\subseteq X$ and $r\in \mathcal{B}(\reals^n)$ for all $r\in R_\text{int}$. We denote by $r_\text{unsafe} = \reals^n\setminus X$ the unsafe region. We consider a set $AP := \{\Prop_1,\dots,\Prop_{|AP|-1}, \Prop_\text{unsafe}\}$ of \emph{atomic propositions}, and associate a subset of atomic propositions to each region $r\in R_\text{int}\cup \{r_\text{unsafe}\}$. We define the \emph{labeling function} $L:\reals^n\rightarrow 2^{AP}$ as the function that maps each state $x\in \reals^n$ to the atomic propositions that are true in the region where $x$ lies, e.g., if we associate $\{\Prop_1\}$ to region $r_1$, we conclude that $\Prop_1$ is \textit{true} at $x$, denoted $\Prop_1 \equiv \top$, if $x\in r_1$. In consequence, each trajectory $\pathX= x_0 \xrightarrow{u_0} \ldots \xrightarrow{u_{K-1}} x_K$
results in the (observation) \emph{trace} 
$\rho = \rho_0\dots\rho_K, $
where $\rho_k := L(x_k)$. 

In order to formally characterize behaviors of System~\eqref{eq:sys}, we use \emph{linear temporal logic over finite traces} (\LTLf) \cite{de2013linear}, which generalizes Boolean logic to temporal behaviors.
An \LTLf property $\varphi$ is a logical formula defined over atomic propositions $AP$ using Boolean connectives  ``negation'' ($\neg$) and ``conjunction'' ($\land$), and the temporal operators ``until'' ($\luntil$) and ``next'' ($\bigcirc$). The syntax of formula $\varphi$ is recursively defined as
%
    %
    \begin{align*}
    \varphi := \top \mid \Prop \mid \neg\varphi \mid \varphi_1\land\varphi_2 \mid \bigcirc\varphi \mid \varphi_1\luntil\varphi_2,
    \end{align*}
    where $\Prop\in AP$ and $\varphi_1,\varphi_2$ are also \LTLf formulas.
The temporal operators ``eventually'' ($\Diamond$) and ``globally'' ($\Box$) are derived from the above syntax as  $\Diamond \varphi := \top\luntil\varphi$ and $\Box \varphi := \neg\Diamond(\neg\varphi)$.
\LTLf formulae are semantically interpreted over finite traces \cite{de2013linear}. We say a trajectory $\pathX$ satisfies a formula $\varphi$, i.e., $\pathX \models \varphi$, if some prefix of its trace $\rho$ satisfies $\varphi$.

%

Our goal is to synthesize a strategy for System~\eqref{eq:sys} to ensure satisfaction of a given \LTLf formula $\varphi$. However, note that (i) under a given strategy, the satisfaction of $\varphi$ is probabilistic, and (ii) in our setting, the distribution of the disturbance is unknown. Hence, we aim to leverage data samples to generate a strategy that guarantees System~\eqref{eq:sys} satisfies $\varphi$ with high probability. Furthermore, note that the synthesized strategy must account for the learning gap due to the lack of knowledge of $P_W$. 

\begin{problem}
    \label{prob:Syntesis}
    Consider stochastic System~\eqref{eq:sys}, a set $\{\hat\w^{(i)}\}_{i=1}^N$ of $N$ i.i.d. samples from $P_W$, a bounded set $X\in\mathcal{B}(\mathbb{R}^n)$ on which Assumption~\ref{ass:lipschitz} holds, and an \LTLf formula $\varphi'$ defined over the regions of interest $R_\text{int}$. Given a confidence level $1-\alpha \in (0,1)$, synthesize a strategy $\sigma_x$ and a high probability bound function $\underline p: X \rightarrow [0, 1]$ such that, with confidence at least $1-\alpha$, for every initial state $x_0\in X$, $\sigma_x$ guarantees that the probability that the paths $\pathX \in \PathX$ satisfy $\varphi := \varphi' \land \Box \neg \Prop_\text{unsafe}$ while remaining in $X$ is lower bounded by $\underline p(x_0)$, i.e.,
    \begin{align*}
    P_{x_0}^{\sigma_x}[\pathX \models \varphi] \ge \underline p(x_0).
    \end{align*}
\end{problem}


We emphasize that the noise distribution $P_W$ is unknown, and no assumptions are imposed on it. Instead, since only samples are available, the probabilistic guarantees for the closed-loop system must hold with a \emph{confidence}. This confidence is related to the probability that the $N$ samples are representative of $P_W$ and is interpreted in the frequentist sense: if the process of obtaining $N$ samples from $P_W$ and synthesizing the strategy is repeated infinitely many times, the condition $P_{x_0}^{\sigma_x}[\pathX \models \varphi] \ge \underline p(x_0)$ for all $x_0 \in X$ holds in at least $1-\alpha$ of the cases.

\paragraph*{Overview of the approach} 

Given the uncountable nature of the state-space of System~\eqref{eq:sys} and the unknown distribution $P_W$, solving Problem \ref{prob:Syntesis} exactly is infeasible. Therefore, we adopt an abstraction-based approach. This method provides a strategy along with a conservative, high-probability bound for every initial state. The abstraction is an uncertain Markov decision process (UMDP) constructed from a finite discretization of set $X$.  We learn the transition relations between the discrete regions using System~\eqref{eq:sys} and disturbance samples, capturing the system's behavior with confidence $1-\alpha$.  
Our UMDP construction is specifically designed to tightly capture the learning uncertainty.
Then, we devise a strategy synthesis algorithm based on robust dynamic programming to (robustly) maximize the probability of satisfying $\varphi$ on this UMDP.  Next, we refine the obtained strategy to System~\eqref{eq:sys} such that it guarantees the closed-loop system satisfies $\varphi$ with a probability higher than the one obtained for the abstraction with confidence $1-\alpha$.

\section{Preliminaries on Uncertain Markov Decision Processes}
\label{sec:preliminaries}


An uncertain MDP (UMDP), also known as a robust MDP, is a stochastic system that generalizes the MDP class by allowing its transition probability distributions to be uncertain, taking values from a set \cite{iyengar2005robust,el2005robust,wiesemann2013robust}.

\begin{definition}[Uncertain MDP]
\label{def:robust_mdp}
    A labeled uncertain Markov decision process (UMDP) $\M$ is a tuple $\M = (S,A,\Gamma,s_0,AP,L)$, where 
    $S$ and $A$ are finite sets of states and actions, respectively, $s_0 \in S$ is the initial state,
    $\Gamma = \{\Gamma_{s,a} \subseteq \mathcal{P}(S) : {s\in S, a\in A}\}$, where $\mathcal{P}(S)$ is the set of probability distributions over $S$, and $\Gamma_{s,a}$ is a nonempty set of transition probability distributions for state $s\in S$ and action $a \in A$,
    $AP$ is a finite set of atomic propositions, and $L:S\rightarrow 2^{AP}$ denotes the labeling function.
\end{definition}

A finite \emph{path} of UMDP $\M$ is a sequence $\pathrmdp = s_0 \xrightarrow{a_0} s_1 \xrightarrow{a_1} \ldots 
\xrightarrow{a_{K-1}} s_K$ of states $s_k\in S$ and actions $a_k \in A$ such that there 
exists $\gamma\in\Gamma_{s_k,a_k}$ with $ \gamma(s_{k+1})> 0$ for all $k \in \{0, \dots, K-1\}$. We denote by $\Pathrmdp$ the set of all finite paths. Given a path $\pathrmdp\in\Pathrmdp$, $\pathrmdp(k) = s_k$ is the state of $\pathrmdp$ at time $k \in \{0,\dots, K\}$, and we denote its last state by $\last(\pathrmdp)$. 
A \emph{strategy} of a UMDP $\M$ is a function $\policy: \Pathrmdp \rightarrow A$ that maps each finite path to the next action. We denote by $\Sigma$ the set of all strategies of $\M$. 
Given path $\pathrmdp \in \Pathrmdp$ and $\policy \in \Sigma$, the process evolves from $s_{k} = \last(\path)$ under $a_k = \policy(\pathrmdp)$ to the next state according to a probability distribution in $\Gamma_{s_k,a_k}$. An adversary is a function that chooses this distribution \cite{givan2000bounded}. Formally, an \emph{adversary} is a function $\adversary: S \times A \times \big(\naturals\cup \{0\}\big) \rightarrow \mathcal{P}(S)$ that maps each state $s_k$, action $a_k$, and time step $k\in\naturals\cup \{0\}$ to a transition probability distribution $\gamma\in\Gamma_{s_k,a_k}$, according to which $s_{k+1}$ is distributed. We denote the set of all adversaries by $\Adversary$. Given an initial condition $s_0\in S$, a strategy $\policy\in\Setofpolicies$ and an adversary $\xi\in\Xi$, the UMDP collapses to a Markov chain
with a unique probability distribution $Pr_{s_0}^{\policy,\xi}$ over its paths.

\begin{definition}[Interval MDP]
\label{def:imdp}
    A labeled interval Markov decision process (IMDP) $\I$ is an UMDP whose transition probability distributions are defined by intervals: for all $s\in S$, $a\in A$, $\Gamma_{s,a} := \{ \gamma \in \mathcal{P}(S) : \underline P(s,a)(s')\le \gamma(s') \le \overline P(s,a)(s') \}$, where $\underline P(s,a)(\cdot), \overline P(s,a)(\cdot) : S \rightarrow [0,1]$.
\end{definition}

\section{Data-driven UMDP Abstraction}
\label{sec:abstraction}
In this section, we introduce a construction of a UMDP $\M$, whose path probabilities are guaranteed to encompass the probabilities of System~\eqref{eq:sys}'s trajectories with confidence $1-\alpha$
. 
%
We define the set of states $S$ of $\M$ as follows. Let $R := \{r_1,\dots, r_{|R|}\}$ be a finite partition of the continuous state-space $\mathbb{R}^n$ into non-overlapping, non-empty regions, which respects the regions of interest $R_\text{int}$ and the safe set $X$, and such that $r\in\mathcal{B}(\reals^n)$ for all $r\in R$. We let region $r_{|R|} := r_\text{unsafe}$ represent the unsafe set. 
We assign each region $r \in R$ to a state $s \in S$ in the abstraction $\M$ through the bijective map $J: R \rightarrow S$, which ensures that $J^{-1}(s)  = r \in R$ is unique. For simplicity, we abuse the notation and also say $J(x) = s$ if $x \in r$ with $J(r) = s$.
We define the action set $A$ of $\M$ to be the finite control set $U$ of System~\eqref{eq:sys}. Furthermore, and with a slight abuse of language, we denote by $L$ the labeling function of $\M$, which maps each state $s\in S$ to the atomic propositions that hold at $x\in r = J^{-1}(s)$. 

Next, we define the set of transition probability distributions of the abstraction. To that end, we begin by stating the following proposition, whose proof follows from \cite[Eq.12]{badings2023probabilities}\footnote{Measurability of the events in  \eqref{eq:tp_bounds} is formally proved in Appendix~\ref{sec:measurability}.}, which gives uniform bounds in the probabilities that System~\eqref{eq:sys} transitions from each point $x$ in some region $r\in R$ to some region $\tilde r \in \mathcal{B}(\reals^n)$.

\begin{proposition}
    \label{prop:tp_bounds}
    Given a region $r\in R$, an action $a\in A$ and a realization $w\in W$ of $\w$, denote by $\text{Reach}(r,a,w) := \{f(x,a,w) : x\in r\}$ the reachable set of $r$ under $a$ and $w$. 
    Then, the probability of transitioning from each state $x\in r$ to region $\tilde r \in \mathcal{B}(\reals^n)$ under action $a\in A$ is bounded by
    \begin{align}
\label{eq:tp_bounds}
    P_W(\{w\in W : \text{Reach}(r,a,w)\subseteq \tilde r \}) \le T(\tilde r \mid x, a) \le P_W(\{w\in W : \text{Reach}(r,a,w) \cap \tilde r \neq \emptyset \})
\end{align}
\end{proposition}
%
%
Below, we use the samples of $\w$ to derive data-driven bounds that contain the ones in \eqref{eq:tp_bounds}, and leverage them to define the set of transition probability distributions for $\M$.

\subsection{Data-Driven Transition Probability Bounds}
\label{sec:data_driven_bounds}

We now construct the sets $\Gamma_{s,a}$ of transition probability distributions of the abstraction by leveraging the samples from $\w$. Specifically, in our UMDP abstraction, the set $\Gamma_{s,a}$ for each state-action pair $(s,a)$ is defined by: (i) interval bounds on the probability of transitioning to each state $s'\in S$, (ii) interval bounds on the probability of transitioning to a cluster of states in $2^S$, and (iii) a bound on the probability of transitioning to states within the reachable set of the learned support of $P_W$.  Notably, (ii) and (iii) distinguish our construction from prior work, which relies solely on (i). As a result, our UMDP incorporates additional constraints, leading to tighter uncertainty sets. This yields less conservative probabilistic guarantees, as shown in the case studies. 



To derive the bounds in steps (i)-(iii), we use Proposition~\ref{prop:tp_bounds}, samples from $\w$, and two well-known concentration inequalities. The proposition below 
enables us to compute bounds on transition probabilities between regions, which we later use to obtain bounds in (i)-(ii).

\begin{proposition}
\label{prop:hoeffding}
    Consider the set $\{\boldsymbol{\hat w}^{(i)}\}_{i=1}^N$ of i.i.d. samples from $\w$. Pick $r\in R$, $a\in A$, $\tilde r \in \mathcal{B}(\reals^n)$ and $\beta \in(0, 1)$, and let $\epsilon = \sqrt{\log(2/\beta)/(2N)} >0 $. Then, with confidence at least $1-\beta$ we have 
    that, for all $x\in r$,
    \begin{subequations}
        \label{eq:data_driven_bounds}
        \begin{align}
            T(\tilde r\mid x, a) &\ge \underline P(r, a)(\tilde r) := \frac{1}{N} \Big| \{i\in\{1,\ldots,N\}:\text{Reach}(r,a, \boldsymbol{\hat w}^{(i)})\subseteq \tilde r \} \Big| - \epsilon \label{eq:data_driven_bounds_lower}\\
            T(\tilde r\mid x, a) &\le \overline P(r, a)(\tilde r) := \frac{1}{N} 
            \Big| \{i\in\{1,\ldots,N\}:\text{Reach}(r,a, \boldsymbol{\hat w}^{(i)})\cap \tilde r \neq \emptyset\} \Big| + \epsilon. \label{eq:data_driven_bounds_upper}
        \end{align}
    \end{subequations}
\end{proposition}
\begin{proof}
Consider the lower bound in \eqref{eq:tp_bounds}. Denote $E := \{w\in W : \text{Reach}(r,a,w)\subseteq \tilde r \}$ and note that
%
    $T(\tilde r\mid x, a) \ge P_W(E) = \mathbb{E}_P[\mathds{1}_{E}(\w)]$
%
for all $x\in r$. Therefore applying Hoeffding's inequality to the random variable $\frac{1}{N}\sum_{i=1}^N\mathds{1}_{E}(\boldsymbol{\hat w}^{(i)})$ yields 
%
    $P_W^N[P_W(E)  \ge \frac{1}{N}\sum_{i=1}^N\mathds{1}_{E}(\boldsymbol{\hat w}^{(i)}) - \epsilon] \ge 1-\beta/2$,
%
with $\epsilon = \sqrt{\log(2/\beta)/(2N)}$. Thus, the first expression in ~\eqref{eq:data_driven_bounds} holds 
for all $x\in r$ with confidence $1-\beta/2$. Employing a similar argument, we obtain that the second expression in ~\eqref{eq:data_driven_bounds} also holds
for all $x\in r$ with confidence $1-\beta/2$. Combining both results via the union bound, we obtain the result.
\end{proof}
%


\noindent

\begin{remark}
    The complexity of computing the bounds in \eqref{eq:data_driven_bounds} is proportional to $N$, which is typically high to obtain tight bounds. To reduce this complexity, we cluster the $N$ samples from $\w$ into $N_c \ll N$ clusters, each with center $c_j$ and diameter $\phi_j$. Substituting the sets $\text{Reach}(r,a,\boldsymbol{\hat w}^{(j)})$ in \eqref{eq:data_driven_bounds} by $\{f(x,a,w) \in \reals^n : x\in r, \|w - c_j\| \le \phi_j/2)\}$, it is evident that Proposition~\ref{prop:hoeffding} still holds, with relaxed bounds. Note that this clustering induces a partition on $W$, allowing to overapproximate the sets $\{f(x,a,w) \in \reals^n : x\in r, \|w - c_j\| \le \phi_j/2)\}$ as shown by \cite{skovbekk2023formal}.
\end{remark}


Next, we estimate the support of $P_W$ in (iii). Including this information into $\M$ tightens the sets $\Gamma_{s,a}$ of transition probability distributions, thus yielding a less conservative abstraction.

\begin{proposition}[Confidence Region]
\label{prop:confidence}
Let $\hat c, \epsilon > 0$ $\beta_c >0$ and 
%
    $N \ge \log{(1/\beta_c)}/\log{(1/(1-\epsilon_c}))$,
%
it holds, with a confidence greater than $1-\beta_c$ with respect to the random choice of $\{\boldsymbol{\hat w}^{(i)}\}_{i = 1}^N$, that 
%
    $P(\{w\in W : \|w\| \le \hat c\}) \ge 1 - \epsilon_c$.
\end{proposition}

%

        \begin{figure}[h]
            \includegraphics[width=0.4\textwidth]{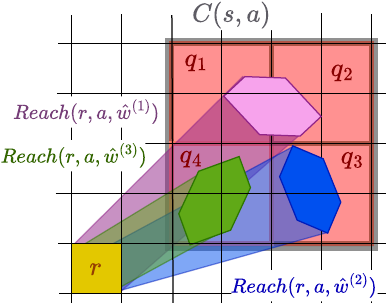}
            \caption{\small Illustration of the sets in Def.~\ref{def:robust_mdp}. $C(s,a) = \{q_1,q_2,q_3,q_4\}$, and each $q_i$ contains $4$ states. The probability that the successor state of $s=J(r)$ under action $a$ will be in $C(s,a)$ is higher than $1-\epsilon$. Note that the reachable sets corresponding to $\boldsymbol{\hat w}^{(2)}$ and $\boldsymbol{\hat w}^{(3)}$ are contained in $q_3$ and $q_4$, respectively, but no single region in the fine partition contains them completely.}
            \label{fig:figure}
        \end{figure}

We denote the learned confidence region for $w$ by $\widehat W := \{w\in W : \|w\| \le \hat c\}$, which contains at least $1-\epsilon_c$ probability mass from $P_W$ with a confidence greater than $1-\beta_c$. We also define $\widehat{\text{Post}}(s,a) := \{J(f(x,a,w))\in S : x\in J^{-1}(s), w\in \widehat W\}$ for each $s\in S$, $a\in A$ as the set of states of $\M$ that can be reached from region $J^{-1}(s)$ and for some disturbance $w\in\widehat W$.

We now have all the components needed to formally define our abstraction class. 
Intuitively, the abstraction relies on a two-layer discretization: a fine one represented by $S$ and a coarse one formed by clustering the elements of $S$ (see Figure~\ref{fig:figure}). 
Let $Q \subseteq 2^S$ represent this clustering, which is non-overlapping, i.e., $\bigcup_{q \in Q} q = S$ and $q \cap q' = \emptyset$ for all  $q \neq q' \in Q$. This clustering is crucial for obtaining non-zero lower-bound transition probabilities in \eqref{eq:data_driven_bounds_lower}, as $\text{Reach}(r, a, \hat{\w}^{(i)})$ often cannot be contained within a single small region but can be captured by a cluster of regions (see Figure~\ref{fig:figure}). 
We highlight that we do not require $R$ and $Q$ having a specific shape, besides the considerations previously mentioned, and that varying $R$ and $Q$ yields a different abstraction. These shapes are therefore hyperparameters of our approach. Additionally, we leverage the learned support $\widehat W$ of the disturbance to impose the constraint that the successor state corresponding to a given state-action pair lies on some region with high probability. As described in Sections~\ref{sec:issues}, \ref{sec:issues2}, this constraint is key to make our approach work in practice. With this intuition, we formally define our abstraction as follows.

\begin{definition}[UMDP Abstraction]
    \label{def:umdp abstraction}
    Let $Q\subseteq 2^S$ be a non-overlapping clustering of $S$ and, for all $s\in S\setminus\{s_{|S|}\}$, $a\in A$, let $Q(s,a) \subseteq Q $ be the subset 
    that covers $\widehat{\text{Post}}(s,a)$, i.e., $\widehat{\text{Post}}(s,a)$ is contained in $C(s,a) := \bigcup_{q\in Q(s,a)} q$. We define the UMDP abstraction of System~\eqref{eq:sys} as $\M = (S,A,s_0,\Gamma,AP,L)$, with, $\forall s\in S\setminus\{s_{|S|}\}$ and $\forall a\in A$,
    \begin{align}
    \label{eq:Gamma}
        \Gamma_{s,a} := \big\{ & \gamma \in \mathcal{P}(S) : \underline P(r_s, a)(r_{s'})  \le \gamma(s') \le \overline P(r_s, a)(r_{s'}) \quad\;\;\; \forall s' \in C(s,a), \nonumber \\
        &\; \underline P(r_s, a)(r_{q'}) \le \sum_{s'\in q'}\gamma(s') \le \overline P(r_s, a)(r_{q'}) \:\; \forall q' \in Q(s,a), \!\!\sum_{s'\in C(s,a)} \!\!\!\! \gamma(s') \ge 1 - \epsilon_c \big\}, 
    \end{align}
where $\underline P, \overline P$ are defined in \eqref{eq:data_driven_bounds},
%
$r_{s'} = J^{-1}(s')$, and $r_{q'} = \bigcup_{s'\in q'} J^{-1}(s')$, and $\Gamma_{s_{|S|},a} = \{\delta_{s_{|S|}}\}$ for all $a\in A$, where $\delta_{s_{|S|}}$ denotes the Dirac measure located at $s_{|S|}$.
\end{definition}

Note that, by making the unsafe state absorbing, we embed the safety part $\Box \neg \Prop_\text{unsafe}$ of $\varphi$ into $\M$, because the only paths of $\M$ that satisfy $\varphi'$ are those that remain in $X$. In Theorem~\ref{theorem:soundness}, we establish that the UMDP $\M$ is a sound abstraction of System~\eqref{eq:sys}, i.e., that $\M$ captures all $1$-step behaviors of System~\eqref{eq:sys}.


\begin{theorem}[Soundness of UMDP Abstraction]
\label{theorem:soundness}
    For all $s\in S\setminus\{s_{|S|}\}$, $a\in A$, $x\in J^{-1}(s)$, $s'\in S$ define $\gamma_x\in\mathcal{P}(S)$ as
    %
        $\gamma_x(s') := T(J^{-1}(s')\mid x,a)$ 
    %
    for all $s'\in S$. Then, $\gamma_x\in \Gamma_{s,a}$ for all $s\in S\setminus\{s_{|S|}\}$, $a\in A$, with confidence of at least $1-\alpha$, where $\alpha = \beta_c + \big( \sum_{s\in S\setminus\{s_{|S|}\}, a\in A} |C(s,a)| + |Q(s,a)| \big)\beta$.
\end{theorem}
%
\begin{proof}
    For clarity, let $r_s := J^{-1}(s)$ for all $s\in S$. 
    By Proposition~\ref{prop:hoeffding} we have that $\underline P(r_s,a)(r_{s'}) \le \gamma_x(s') \le \overline P(r_s,a)(r_{s'})$ with confidence $1-\beta$ pointwisely for all $s\in S\setminus\{s_{|S|}\}$, $a\in A$, $s'\in C(s,a)$. Note that, for all $s' \notin C(s,a)$, no interval constraints are learned. 
    Therefore, the total number of learned intervals for states $s'$ is $\sum_{s\in S\setminus\{s_{|S|}\}, a\in A} |C(s,a)|$. Using again Proposition~\ref{prop:hoeffding} with $\tilde r = \bigcup_{s'\in q'}r_{s'}$ and noting that $T(\tilde r \mid x, a) = \sum_{s'\in q'} T(r_{s'} \mid x, a)$ since the regions in the partition $R$ are disjoint, it follows that $\underline P(r_s,a)(r_{q'}) \le \sum_{s'\in q'}\gamma_x(s') \le \overline P(r_s,a)(r_{q'})$ with confidence $1-\beta$ pointwisely for all $q'\in Q(s,a)$, $s\in S\setminus\{s_{|S|}\}$, $a\in A$. This makes the total number of learned intervals for clusters $q'$ be $\sum_{s\in S\setminus\{s_{|S|}\}, a\in A} |Q(s,a)|$.
    Furthermore, since, by definition, $C(s,a) \supseteq \widehat{\text{Post}}(s,a)$, we also have that
    %
    %
    $\sum_{s'\in C(s,a)}\gamma_x(s') \ge$ $\sum_{s'\in \widehat{\text{Post}}(s,a)}\gamma_x(s') = \sum_{s'\in \widehat{\text{Post}}(s,a)}T(r_{s'}\mid x,a)
    = T(\bigcup_{s'\in \widehat{\text{Post}}(s,a)} r_{s'} \mid x,a)
    = P_W(\{w\in W : f(x,a,w) \in \bigcup_{s'\in \widehat{\text{Post}}(s,a)} r_{s'}\}) = P_W(\{w\in W : J(f(x,a,w)) \in \widehat{\text{Post}}(s,a)\}) \ge P_W(\widehat W),$
    which is at least $1-\epsilon_c$ with confidence at least $1-\beta_c$. Then, it follows that the last constraint in \eqref{eq:Gamma} also holds with the same confidence. Combining this confidence with the $\sum_{s\in S\setminus\{s_{|S|}\}, a\in A} |C(s,a)| + |Q(s,a)|$ learned intervals via the union bound, we get that, with confidence at least $1-\alpha$, $\gamma_x$ fulfills all constraints in the definition of $\Gamma_{s,a}$ for all $s\in S\setminus\{s_{|S|}\}$, $a\in A$.
\end{proof}
Note that $|C(s,a)|$ and $|Q(s,a)|$ depend on the samples from $P_W$ trough $\widehat W$ (see Proposition~\ref{prop:confidence}). However, the expression for $\alpha$ in Theorem~\ref{theorem:soundness} follows from a union bound argument, which holds only if the number of transitions $\sum_{s\in S\setminus\{s_{|S|}\}, a\in A} |C(s,a)| + |Q(s,a)|$ is deterministic. We solve this issue by first estimating this number. Then, given $\alpha$ and $\beta_c$ and $\epsilon_c, N$, we solve for $\beta,\epsilon$, and obtain the corresponding abstraction $\M$. Finally, we check that the actual number of transitions does not exceed the estimated one, guaranteeing that $\M$ is sound by Theorem~\ref{theorem:soundness}.

\begin{corollary}
    \label{cor:sample_complexity}
    Given $\alpha \in (0,1)$ and $\epsilon, \epsilon_c \in (0,1)$, the sample complexity of obtaining a UMDP abstraction with confidence at least $1-\alpha$ is
        $N = \max\big\{\log( \frac{n_\text{learn}/\alpha}{2\epsilon^2}), \log(\frac{n_\text{learn}/\alpha}{\log{(1/(1-\epsilon_c}))}) \big\}$,
    %
    with $n_\text{learn} = 1+ \sum_{s\in S\setminus\{s_{|S|}\}, a\in A} |C(s,a)| + |Q(s,a)|$.
\end{corollary}



\subsection{Issues of Na\"ive Data-Driven IMDP Abstractions: Loose Abstraction}
\label{sec:issues}

As mentioned above, compared to the abstraction classes frequently used in the literature, namely, interval MDPs (IMDPs) \cite{lahijanian2015formal}, our UMDP abstraction differs in that the sets $\Gamma_{s,a}$ of transition probability distributions are defined by more constraints than just the intervals $[\underline P(r_s,a)(r_{s'}), \overline P(r_s,a)(r_{s'})]$ for all $s'\in S$. Here, we discuss why IMDPs are not a good abstraction choice in our setting because they do not capture the dynamics~\eqref{eq:sys} very tightly, and thus why a more complex $\Gamma$ is required. Then, in Section~\ref{sec:issues2} we explain how the strategy synthesis process often fails to return meaningful results for these na\"ive abstractions.

Let $\mathcal{I}$ be a UMDP abstraction of System~\eqref{eq:sys}, where $\Gamma$ is defined only by the first row of constraints in \eqref{eq:Gamma}. 
Note that the lower bound $\underline P(r_s,a)(r_{s'})$ is obtained using Expression~\eqref{eq:data_driven_bounds_lower}, which boils down to checking wether or not the the event $\text{Reach}(r_s,a,\boldsymbol{\hat w}^{(i)})\subseteq r_{s'}$ happens. However, this condition rarely takes place, since it requires the reachable set being smaller than the region $r_{s'}$ (see Figure~\ref{fig:figure}), which is uncommon in a big portion of the state-space unless (i) the system's dynamics are $1$-step contractive and (ii) the partition is "aligned" with the dynamics. As a consequence, when $s'$ is not a terminal state, e.g., the goal or unsafe regions, which are typically bigger than the rest, it very often happens that
$\underline P(r_s,a)(r_{s'}) = 0$ for all such states $s'$. However, the upper bound $\overline P(r_s,a)(r_{s'})$ behaves very differently: by \eqref{eq:data_driven_bounds_upper}, $\overline P(r_s,a)(r_{s'}) \ge \epsilon$ for all $s'\in S$. As a consequence, the set $\Gamma_{s,a}$ contains many spurious distributions, making $\I$ a very loose abstraction of System~\eqref{eq:sys}. In Section~\ref{sec:issues2} we describe how this looseness of $\I$ typically translates into poor strategy synthesis results.

\section{Strategy Synthesis}
\label{sec:synthesis}

%
%


Here, we focus on synthesizing a strategy for System~\eqref{eq:sys} and provide a lower bound on the probability that the closed-loop system satisfies the \LTLf formula $\varphi$. We first show that standard synthesis procedures for general UMDP abstractions from the literature \cite{wolff2012robust, gracia2024data,gracia2025efficient}
also apply to our setting and then introduce a novel (tailored) algorithm that leverages the specific structure of our UMDP abstraction to reduce computational complexity. 

Strategy synthesis begins by translating $\varphi$ into its equivalent deterministic finite automaton (DFA) $\mathcal{A}^\varphi$ \cite{de2013linear} and constructing the product $\M^\varphi = \M \otimes \mathcal{A}^\varphi$. A strategy $\sigma^\varphi$ is then synthesized on $\M^\varphi$ via unbounded-horizon \emph{robust dynamic programming} (RDP) with a reachability objective \cite{wolff2012robust}, as detailed in \cite[Theorems 6.2, 6.6]{gracia2025efficient}. $\sigma^\varphi$ robustly maximizes the probability of satisfying $\varphi$ under adversarial choices of transition probabilities from the set $\Gamma^\varphi$. Finally, $\sigma^\varphi$ is refined into a strategy $\sigma_x$ for System~\eqref{eq:sys}.

We start by defining the DFA $\A^\varphi$.
\begin{definition}[DFA]
    Let $\varphi$ be an \LTLf formula defined over a set of atomic propositions $AP$. The \emph{deterministic finite automaton} (DFA) corresponding to $\varphi$ is a tuple $\A^\varphi = (Z,2^{AP},\delta,z_0,Z_F)$ where $Z$ is a finite set of states, $2^{AP}$ is a finite set of input symbols, $\delta:Z\times 2^{AP} \rightarrow Z$ is the transition function, $z_0\in Z$ is the initial state, and $Z_F\subseteq Z$ is the set of accepting states.
\end{definition}

Given a trace $\rho = \rho_0\rho_1 \dots \rho_{K} \in (2^{AP})^*$, a run $z = z_0z_1 \dots z_{K+1}$ is induced on $\A^\varphi$, where $z_{k+1} = \delta(z_k,\rho_{k})$ for all $k\in \{0,\dots,K\}$. By construction of $\A^\varphi$, trace $\rho$ satisfies $\varphi$ iff $z_{K+1} \in Z_F$ \cite{de2013linear}. Such a run is called \emph{accepting} for $\A^\varphi$. 

Next, we define the product UMDP $\M^\varphi$, which contains information about the (uncertain) transition probabilities of $\M$ and the transition function of $\A^\varphi$.
\begin{definition}[Product UMDP]
    Given UMDP $\M$ and DFA $\A^\varphi$, the product $\M\otimes\A^\varphi$ is another UMDP $\M^\varphi = (S^\varphi,A^\varphi,\Gamma^\varphi,s_0^\varphi, S_{F}^\varphi)$, where
   $S^\varphi = S\times Z$ 
   , $A^\varphi = A$ 
   , $s_0^\varphi = (s_0,  \delta(z_0,L(s_0)) )$, 
   $S_F^\varphi = S\times Z_F$ 
   ,
   and 
    $\Gamma^\varphi = \{\Gamma_{(s,z),a}^\varphi: (s,z)\in S^\varphi, a\in A^\varphi\}$ with 
       $
            \Gamma_{(s,z),a}^\varphi := \{\gamma^\varphi \in \mathcal{P}(S^\varphi) : \exists \gamma\in\Gamma_{s,a}\: \text{s.t.}\: \gamma^\varphi((s',z')) = \gamma(s') \: \text{with}\: z' = \delta(z,L(s')), \forall s'\in S\}.
       $ We denote a finite path of $\M^\varphi$ by $\omega^\varphi$ and the set of all such paths by $\Omega^\varphi$. We also let $\Sigma^\varphi$ and $\Xi^\varphi$ denote the sets of strategies and adversaries of $\M^\varphi$, respectively.
\end{definition}
Intuitively, $\M^\varphi$ is a UMDP whose state is the product between the state spaces of $\M$ and $\A^\varphi$, and whose transition probability distributions combine information regarding the transitions of $\M$ and $\A^\varphi$. Specifically, the set $\Gamma_{(s,z),a}^\varphi$ corresponding to a given state-action pair $((s,z),a)$ contains probability distributions $\gamma^\varphi$ over the product space $S^\varphi$ such that their projections (pushforward measure) $\gamma$ onto the set of probability distributions $\mathcal{P}(S)$ belong to $\Gamma_{s,a}$. Conversely, each $\gamma^\varphi\in\Gamma_{(s,z),a}^\varphi$ is obtained by taking some $\gamma\in\Gamma_{s,a}$ and lifting it to the space $\mathcal{P}(S^\varphi)$ by taking into account the transition function of $\A^\varphi$.

Having obtained the product UMDP $\M^\varphi$, we synthesize a strategy $\sigma^\varphi$ which maximizes the probability of reaching set $S_F^\varphi$ under an adversarial choice of the transition probabilities of $\M^\varphi$ by the adversary. It can be proved \cite{wolff2012robust} that this strategy, when mapped to a strategy $\sigma\in\Sigma$ of $\M$, also maximizes the worst-case probability of $\M$ satisfying the \LTLf formula $\varphi$. 
Proposition~\ref{prop:robust_value_iteration} shows how to obtain the reachability probabilities.
\begin{proposition}[Robust Dynamic Programming {\cite[Theorem 6.2]{gracia2025efficient}}]
\label{prop:robust_value_iteration}
Given $s^\varphi\in S^\varphi$, define the optimal robust reachability probability as 
\begin{align}
\label{eq:max_reachability_probability}
\underline p(s^\varphi) :=\sup_{\sigma^\varphi\in\Sigma^\varphi}\inf_{\xi^\varphi\in\Xi^\varphi} Pr_{s^\varphi}^{\policy^\varphi,\xi^\varphi}(\{\omega^\varphi\in\Omega^\varphi : \exists k\in\naturals\cup \{0\} \:\text{s.t.}\: \omega^\varphi(k)\in S_F^\varphi\}).
\end{align}
Consider also the recursion
\begin{align}
\label{eq:robust_value_iteration_lower_bound}
    \underline p^{k+1}(s^\varphi) = 
        \begin{cases}
             1 &\text{if }\:s^\varphi \in S_F^\varphi\\ \max\limits_{a\in 
             A^\varphi}\min\limits_{\gamma\in\Gamma_{s^\varphi,a}^\varphi} \sum\limits_{s^{\varphi\prime}\in 
             S^\varphi}\gamma(s^{\varphi\prime})\underline p^{k}(s^{\varphi\prime}) &\text{otherwise},
             \end{cases}
\end{align}
for all $k\in\naturals\cup\{0,\infty\}$, with initial condition $\underline p^{0}(s^\varphi) = 1$ for all $s^\varphi\in S_F^\varphi$ and $0$ otherwise. Then, recursion \eqref{eq:robust_value_iteration_lower_bound} converges pointwisely to $\underline p$
.
\end{proposition}
Having obtained the reachability probabilities, we obtain a memoryless strategy $\sigma^\varphi$ that attains the probability in \eqref{eq:max_reachability_probability} for all states by following the procedure in \cite[Theorem 6.6]{gracia2025efficient}, and then map it to a finite-memory strategy $\sigma_x$ of System~\eqref{eq:sys} as described in \cite[Section 6.4]{gracia2025efficient}. The following theorem 
ensures that satisfaction probability bounds are preserved under this procedure, thus solving Problem~\ref{prob:Syntesis}. 


\begin{theorem}[Strategy Synthesis through Product UMDP]
\label{thm:strategy_synthesis}
    Let $\sigma^\varphi$ and $\underline p^\varphi$ be respectively the optimal strategy and the lower bound in the probability of satisfying $\varphi$ obtained via RDP on $\M^\varphi$. 
    Furthermore, let $\sigma_x$ be the strategy obtained by refining $\sigma^\varphi$ to System~\eqref{eq:sys}. Then, with confidence $1-\alpha$,
    %
    %
    $Pr_{x}^{\sigma^\varphi}[\pathX \models \varphi] \ge \underline p^\varphi((s,z))$
    for all $x\in X$, where $s = J(x)$, $z = \delta(z_0,L(s))$.
\end{theorem}
\begin{proof}
    Assume $\M$ is a correct abstraction of System~\eqref{eq:sys}. Then, it follows that \cite[Theorem 2]{jackson2021formal} $Pr_{x}^{\sigma^\varphi}[\pathX \models \varphi] \ge \underline p^\varphi((s,z)),$
    for all $x\in X$, where $s = J(x)$, $z = \delta(z_0,L(s))$. Since by Theorem \ref{theorem:soundness} this assumption holds with confidence not smaller than $1-\alpha$, the implication also holds with the same confidence, which concludes the proof.
\end{proof}

\subsection{Issues of Na\"ive Data-Driven IMDP Abstractions: Overly Conservative Solution to Problem~\ref{prob:Syntesis}}
\label{sec:issues2}

In this subsection we show that strategy synthesis often yields poor results if the abstraction is a na\"ive IMDP obtained as described in Section~\ref{sec:issues}.

Consider the IMDP abstraction $\I$ of Section~\ref{sec:issues}, and the strategy synthesis process for a simple reachability specification, which is carried out by applying Proposition~\ref{prop:robust_value_iteration} on $\I$ with goal set $S_\text{goal}$. For the sake of clarity, we depict this setup in Figure~\ref{fig:IMDP_example}(a). At iteration $k$, denote by $S_k^0$ the set of states with zero value function (black states in Figure~\ref{fig:IMDP_example}(a)). 
Since Recursion~\eqref{eq:robust_value_iteration_lower_bound} starts with an initial value function that is zero for all $s\notin S_\text{goal}$, it easy to check that $|S_k^0| \approx |S|$ during the first iterations. Due to the form of the transition probability bounds of $\I$, descried in Section~\ref{sec:issues}, 
the adversary $\xi$ is allowed to pick a distribution $\gamma$ over $S$ that assigns probability at least $\epsilon$ of transitioning to each $s'\in S_k^0$, and that, if $S_k^0$ is big enough, the total probability of transitioning to these states adds up to $1$, and thus $s'\in S_k^0$ with probability one. As a result, the value function at the iteration $k+1$ is again zero for all $s\in S_k^0$, and therefore RDP terminates, yielding a vacuous satisfaction probability on a big region of the state space. We empirically demonstrate this issue on case study $\#4$ in Table~\ref{tab:case_studies}, whose results we plot in Figure~\ref{fig:IMDP_example}.

\begin{figure}[b]
    \begin{minipage}[b]{.3\linewidth}
        \centering
        \includegraphics[width=\linewidth]{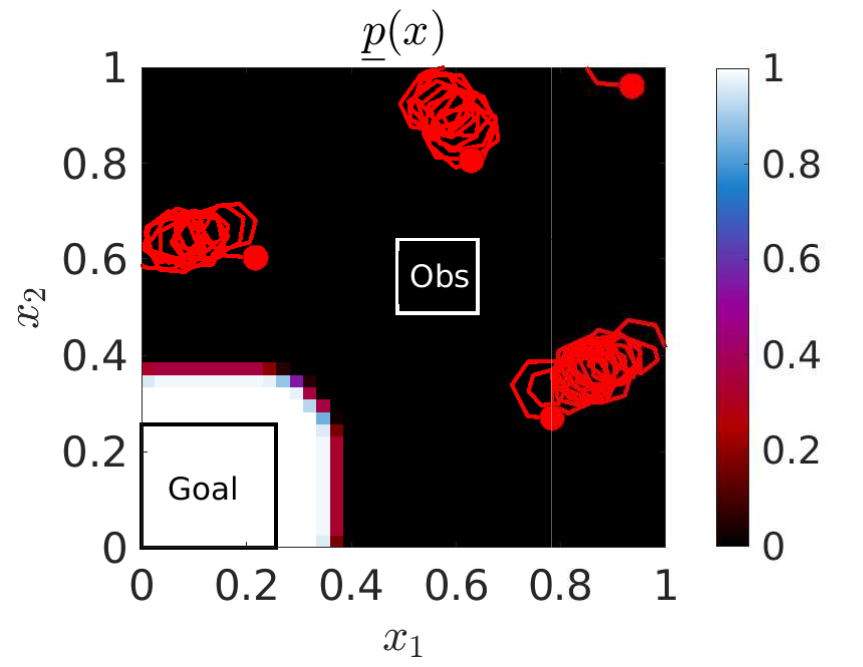}
        \subcaption{Na\"ive IMDP}
    \end{minipage}
    \hfill
    \begin{minipage}[b]{.3\linewidth}
    \centering
        \includegraphics[width=\linewidth]{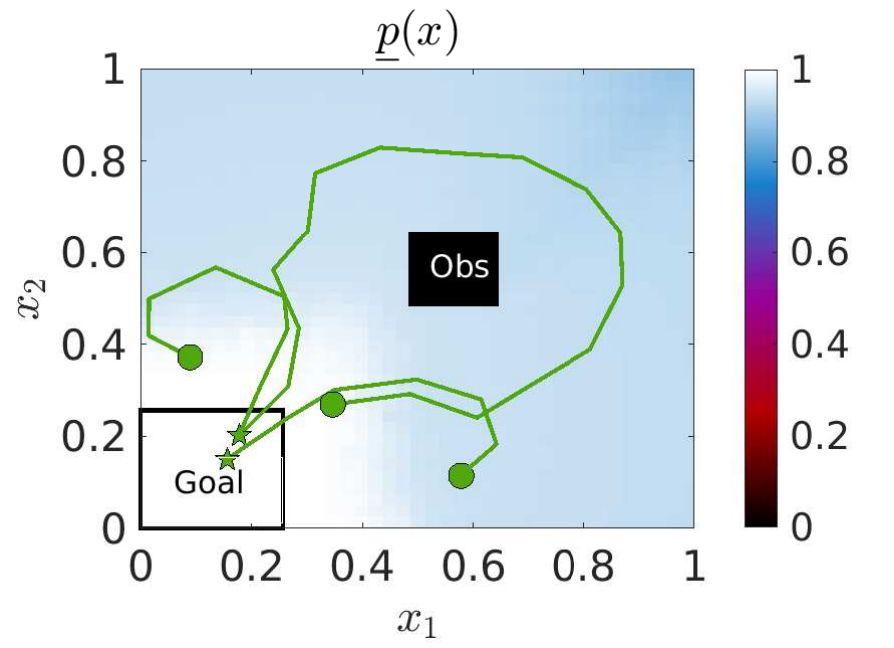}
        \subcaption{IMDP with learned support}
    \end{minipage}
    \hfill
    \begin{minipage}[b]{.3\linewidth}
    \centering
        \includegraphics[width=\linewidth]{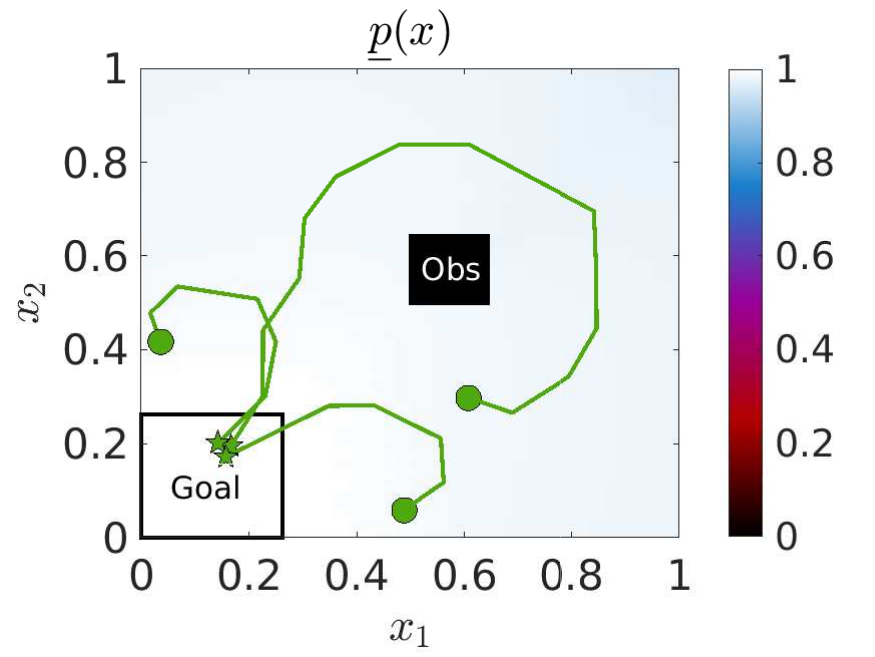}
        \subcaption{UMDP \eqref{def:umdp abstraction}}
    \end{minipage}
    \caption{Empirical demonstration of the issues of na\"ve IMDP abstractions described in Sections~\ref{sec:issues} and \ref{sec:issues2} The figures show the probabilistic guarantees that the unicycle system in case study $\#4$ of Table~\ref{tab:case_studies} satisfies the reach-avoid specification $\varphi_1$ as a function of the initial position and the choice of the abstraction, for the most favorable heading angle. The abstraction choices are:  (a) a na\"ive IMDP, (b) an IMDP with an additional constraint about the learned support of $P_W$ or, equivalently, a UMDP without the constraints that involve clusters in \eqref{eq:Gamma}, (c) a UMDP as in Definition~\ref{def:umdp abstraction}. In all cases, the number of samples used to construct the abstraction was of $N=10^6$. Furthermore, increasing the number of samples to $10^7$ sowed almost identical results when using a na\"ive IMDP abstraction. The green and red lines are trajectories of the unicycle system in closed loop with the synthesized strategy: in green, the ones that satisfy the specification. In red the ones that do not.}
    \label{fig:IMDP_example}
\end{figure}

Note that irrespective of how small $\epsilon>0$ is, refining the abstraction eventually leads to this issue, as it increases the number of states of the abstraction. Furthermore, given a constant discretization granularity, the size of $S_k^0$ is exponential in the dimension of System~\eqref{eq:sys}, which implies that the value of $\epsilon$ required to avoid this issue typically requires an impractical number $N$ of samples from $\w$.

\subsection{Tailored Synthesis Algorithm}

\begin{algorithm}[t]\small
 \caption{$2$-layer $O$-maximization}\label{alg:2O}
 \vspace{-4mm}
\begin{multicols}{2}
\begin{algorithmic}[1]
    \Require $\M,s\in S\setminus\{s_{|S|}\}, a\in A, \underline p^k$
    \Ensure $\gamma$
    \State Sort $\text{Post}(s,a)$ according to $\{\underline p^k(s')\}_{s'\in \text{Post}(s,a)}$ in increasing order
    \State $\gamma(s') \gets \underline P(s,a)(s')$ for all $s'\in \text{Post}(s,a)$
    \State $\gamma(s') \gets 0$ for all $s'\notin \text{Post}(s,a)$
    \State $\gamma(q') \gets \sum_{s'\in q'} \gamma(s')$ for all $q'\in Q(s,a)$
    \State $M \gets 1 - \sum_{s'\in S} \gamma(s')$
    \For{$q'\in Q(s,a)$}
        \State $m \gets \underline P(s,a)(q') - \gamma(q')$
        \If{$m >0$}
            \For{$s'\in q'$}
                \State $g \gets \min\{m, \overline P(s,a)(s') - \gamma(s')\}$
                \State $\gamma(s') \gets \gamma(s') + g$, $\gamma(q') \gets \gamma(q') + g$
                \State $m \gets m - g$, $M \gets M - g$
            \EndFor
        \EndIf
    \EndFor
    \For{$s'\in \text{Post}(s,a)$}
    \State $q' \gets$ get cluster $q'$ such that $s'\in q'$
    \If{$q' \neq \emptyset$}
    \State $m \gets \min\{M, \overline P(s,a)(q') - \gamma(q'),$ $ \qquad \overline P(s,a)(s') - \gamma(s')\}$
    \State $\gamma(q') \gets \gamma(q') + m$
    \Else
    \State $m \gets \min\{M, \overline P(s,a)(s') - \gamma(s')\}$
    \EndIf
    \State $\gamma(s') \gets \gamma(s') + m$, $M \gets M - m$
    \EndFor
\end{algorithmic}
\end{multicols}
\vspace{-4mm}
\end{algorithm}

We introduce a synthesis algorithm tailored for UMDPs as such in Definition~\ref{def:umdp abstraction}, which exploits their structure for greater efficiency.  The algorithm 
draws inspirations from IMDP value iteration \cite{givan2000bounded} to speed up the computation of the optimal adversaries in RDP, i.e., the inner minimization problem in Equation~\eqref{eq:robust_value_iteration_lower_bound}, which is typically formulated a linear program in standard UMDPs. We note that this approach is applicable to our product UMDP $\M^\varphi$ because it retains the same structure of $\M$ 
In consequence and, for simplicity, we then describe the algorithm in the context of a reachability problem on $\M$ rather than on $\M^\varphi$.

Proposition~\ref{prop:product_bounds} shows that $\M$ and $\M^\varphi$ have the same structure.

\begin{proposition}
\label{prop:product_bounds}
    For each $s^\varphi := (s,z)\in (S\setminus\{s_{|S|}\})\times Z$, $a\in A^\varphi$, $q'\in Q(s,a)$, let $q^{\varphi\prime} := \{(s',z')\in q'\times Z : z' = \delta(z, L(s'))\}$, and define $Q^\varphi(s^\varphi,a)$ as the set of all these $q^{\varphi\prime}$. 
    Then, $\Gamma^\varphi_{s^\varphi,a}$ is equivalently expressed as
    \begin{align}
        \nonumber
        \Gamma_{s^\varphi,a}^\varphi = \Big\{ \gamma^\varphi \in & \mathcal{P}(S^\varphi) :\\
        \nonumber
        &\underline P(s^\varphi, a)((s',z'))  \le \gamma^\varphi((s',z')) \le \overline P(s^\varphi, a)((s',z')) \quad \forall s' \in C(s,a), z' \in Z,\\
        \nonumber
        &\underline P(s^\varphi, a)(q^{\varphi\prime}) \le \sum_{s^{\varphi\prime}\in q^{\varphi\prime}} \gamma^\varphi(s^{\varphi\prime}) \le \overline P(s^\varphi, a)(q^{\varphi\prime}) \qquad \;\;\; \forall q^{\varphi\prime} \in Q^\varphi(s^\varphi,a),\\
    \label{eq:Gamma_product}
        &\sum_{s^{\varphi\prime}\in C(s,a)\times Z} \gamma^\varphi(s^{\varphi\prime}) \ge 1 - \epsilon_c \Big\}, 
    \end{align}
    with $\underline{P}(s^\varphi,a)((s',z')) :=
            \underline{P}(r_s,a)(r_{s'})$, $\overline{P}(s^\varphi,a)((s',z')) :=
            \overline{P}(r_s,a)(r_{s'})$ if $z' = \delta(z,L(s'))$ and $0$ otherwise, and $\underline{P}(s^\varphi,a)(q^{\varphi\prime}) :=
            \underline{P}(r_s,a)(r_{q'})$, $\overline{P}(s^\varphi,a)(q^{\varphi\prime}) :=
            \overline{P}(r_s,a)(r_{q'})$, where $q'\in Q(s,a)$ is the projection of $q^{\varphi\prime}$ onto $S$.
    %
%
for all $q^{\varphi\prime} \in Q^\varphi(s^\varphi,a)$. Additionally, for all $z\in Z$ and $a\in A$, $\Gamma^\varphi_{(s_{|S|},z),a} = \{\delta_{(s_{|S|},z')}\}$, with $z' = \delta(z,L(s_{|S|}))$.
\end{proposition}
\begin{proof}
    Pick arbitrary $s^\varphi = (s, z)\in (S\setminus\{s_{|S|}\})\times Z$, $a\in A^\varphi$ and $\gamma^\varphi\in\mathcal{P}(S^\varphi)$ that satisfies the bounds in Expression~\eqref{eq:Gamma_product}. By the interval bounds in said expression for each $(s',z')\in C(s,a)\times Z$, we have that $\gamma^\varphi((s', z')) = 0$ for all $z' \neq \delta(z, L(s'))$. Define $\gamma\in\mathcal{P}(S)$ as $\gamma(s') := \gamma^\varphi(s', \delta(z,L(s')))$ for all $s'\in S$, noting that $\gamma$ is a probability over $S$ because $\sum_{s'\in S}\gamma(s') = \sum_{s'\in S}\gamma^\varphi((s',\delta(z, L(s')))) = \sum_{(s',z')\in S^\varphi}\gamma^\varphi((s',z')) = 1$. Since for each $(s',z')\in S^\varphi$ we have that $\underline{P}(s^\varphi, a)((s',z')) \le \gamma^\varphi((s',z'))\le \overline{P}(s^\varphi, a)((s',z'))$, it is easy to see that $\underline{P}(s,a)(s') \le \gamma(s') \le \overline{P}(s,a)(s')$ for all $s'\in S$. Next, pick $q'\in Q(s,a)$ and let $q^{\varphi\prime}$ be the set in $Q^\varphi(s^\varphi, a)$ such that $q'$ is its projection onto $S$. Because of this relationship we obtain that $\sum_{s'\in q'} \gamma(s') = \sum_{s'\in q'} \gamma^\varphi((s', \delta(z, L(s'))) = \sum_{(s',z')\in q^{\varphi\prime}} \gamma^\varphi((s',z'))\in[\underline{P}((s,z),a)(q^{\varphi\prime}), \overline{P}((s,z),a)(q^{\varphi\prime})] = [\underline{P}(s,a)(q'), \overline{P}(s,a)(q')]$. Finally, since $\gamma^\varphi$ satisfies the last condition in \eqref{eq:Gamma_product}, it follows that
    \begin{align*}
        \sum_{s'\in C(s,a)} \gamma(s') = \sum_{s'\in C(s,a)} \gamma^\varphi((s',\delta(z, L(s')))) = \sum_{s^{\varphi\prime}\in C(s,a)\times Z} \gamma^\varphi(s^{\varphi\prime}) \ge 1 - \epsilon_c.
    \end{align*}
    Therefore, $\gamma\in\Gamma_{s,a}$, which implies that $\gamma^\varphi \in \Gamma_{s^\varphi,a}^\varphi$.

    Conversely, pick a $\gamma^\varphi \in \Gamma_{s^\varphi,a}^\varphi$, and define $\gamma\in\mathcal{P}(S)$ as before. Reversing the arguments in the first part of this proof, we obtain that all $\gamma^\varphi$ satisfies the conditions in Expression~\eqref{eq:Gamma_product}. Since $\gamma^\varphi$ was arbitrary, this holds for all $\gamma^\varphi \in \Gamma_{s^\varphi,a}^\varphi$.  Additionally, for the case that $s = s_{|S|}$, it is easy to check that $\Gamma_{(s,z),a}$ necessarily has the form in the statement of Proposition~\ref{prop:product_bounds}, hence the proof is complete.
\end{proof}



As explained before, strategy synthesis amounts to solving a reachability problem on $\M^\varphi$ and. Since the set $\Gamma^\varphi$ in \eqref{eq:Gamma_product} possesses the same structure as the set $\Gamma$ in \eqref{eq:Gamma}, for the sake of clarity, we now consider simply a reachability problem on $\M$ instead of $\M^\varphi$. Let us first simplify the structure of the sets $\Gamma_{s,a}$ of $\M$ by discarding the last constraint in \eqref{eq:Gamma} and adjusting the transition probability bounds accordingly. Lemma~\ref{lemma:simplify} states that performing RDP on this simplified UMDP yields a lower bound on the reachability probability in \eqref{eq:max_reachability_probability}.


\begin{lemma}
\label{lemma:simplify}
    Consider the UMDP $\M$ as in Definition~\ref{def:umdp abstraction}, and let $\widetilde{\M}$ be the UMDP that differs from the latter only in the sets $\widetilde{\Gamma}_{s,a}$, defined as
    \begin{align}
    \label{eq:Gamma_tilde}
        \widetilde{\Gamma}_{s,a} := \big\{ & \gamma \in \mathcal{P}(S) : \underline P(s, a)(s')  \le \gamma(s') \le \overline P(s, a)(s') \quad\;\;\; \forall s' \in S, \nonumber \\
        &\; \underline P(s, a)(q') \le \sum_{s'\in q'}\gamma(s') \le \overline P(s, a)(q') \:\; \forall q' \in Q(s,a) \big\}, 
    \end{align}
with
    \begin{align*}
        \underline{P}(s,a)(s') &= \begin{cases}
            \underline{P}(r_{s},a)(r_{s'}) \:\quad &\forall s' \in C(s,a)\\
            0  \:\quad &\forall s' \notin C(s,a)
        \end{cases} 
        \\
        \overline{P}(s,a)(s') &= \begin{cases}
            \overline{P}(r_s,a)(r_{s'}) \:\quad &\forall s' \in C(s,a)\setminus \{s_{|S|}\}\\
            0  \:\quad &\forall s' \notin C(s,a)\bigcup \{s_{|S|}\}\\
            \overline{P}(r_s,a)(r_{s'}) + \epsilon_c  \: \quad &\text{if}\: s' = s_{|S|}\:\text{and}\:s' \in C(s,a)\\
            \epsilon_c  \: \quad &\text{if}\: s' = s_{|S|}\:\text{and}\:s' \notin C(s,a)\\
        \end{cases} 
        \\
        \underline{P}(s,a)(q') &= \underline{P}(r_s,a)(r_{q'}) \:\:\qquad\qquad \text{for all}\:\:q' \in Q(s,a)\\
        \overline{P}(s,a)(q') &= \begin{cases}
            \overline{P}(r_s,a)(r_{q'})  \: \quad &\text{if}\:q' \in Q(s,a)\:\text{and}\:\: s_{|S|} \notin q'\\
            \overline{P}(r_s,a)(r_{q'}) + \epsilon_c  \: \quad &\text{if}\:q' \in Q(s,a)\:\text{and}\:\: s_{|S|} \in q'\\
        \end{cases} 
    \end{align*}
    for all $s\in S\setminus\{s_{|S|}\}$, $a\in A$, and $\Gamma_{s_{|S|},a} = \{\delta_{s_{|S|}}\}$ for all $a\in A$. Let $\underbar p_{\widetilde{\M}}$ and $\underbar p$ be the value functions returned by running the RDP recursion in \eqref{eq:robust_value_iteration_lower_bound} on $\widetilde{\M}$ and $\M$, respectively. Then, $\underbar p_{\widetilde{\M}}(s) \le \underbar  p(s)$ for all $s\in S$.
    \end{lemma}
    \begin{proof}
        Consider the value functions $\underline p^k$ and $\underline p_{\widetilde{\M}}^k$ obtained after $k$ iterations of robust dynamic programming (see Proposition~\ref{prop:robust_value_iteration}) on $\M$ and on $\widetilde{\M}$, respectively, and let $s\in S\setminus\{s_{|S|}\}$ and $a\in A$.
        It is easy to observe that a minimizer $\gamma\in\Gamma_{s,a}$ of the inner problem in Expression~\eqref{eq:robust_value_iteration_lower_bound} is any $\gamma$ that assigns as much probability mass as possible to the states $s'\in S$ with the smallest $\underline p^k(s')$. Note that, by the last constraint in the definition of $\Gamma_{s,a}$, at least $1-\epsilon_c$ mass must be allocated to the states in $C(s,a)$, which allows us to define $\gamma\in\mathcal{P}(S)$ with $\tilde{\gamma}(s') := \gamma(s')$ for all $s'\in C(s,a)\setminus\{s_{|S|}\}$, $\tilde{\gamma}(s_{|S|}) := \gamma(s_{|S|}) + \epsilon_c$ and $\tilde{\gamma}(s') := 0$ everywhere else. Then, it is easy to observe that $\tilde{\gamma}\in\widetilde{\Gamma}_{s,a}$. Furthermore, since $\underline{p}^k(s')\ge 0$ for all $s'\in S$ and $\underline{p}^k(s_{|S|}) = 0$, we have that $\sum_{s'\in S}\gamma(s') p^k(s')  \ge \sum_{s'\in S}\tilde{\gamma}(s') p^k(s')$ which, in turn, implies that
        \begin{align*}
            \min_{\gamma\in\Gamma_{s,a}}\sum_{s'\in S}\gamma(s') p^k(s') = \sum_{s'\in S}\gamma(s') p^k(s') \ge \sum_{s'\in S}\tilde{\gamma}(s') p^k(s') \ge \min_{\gamma\in\widetilde{\Gamma}_{s,a}}\sum_{s'\in S}\gamma(s') p^k(s').
        \end{align*}
        Maximizing over the actions on both sides yields $\underline p_{\widetilde{\M}}^{k+1}(s) \le \underline p^{k+1}(s)$, which holds for all $s\in S\setminus\{s_{|S|}\}$. Additionally, it is trivial to check that the previous result also holds for $s = s_{|S|}$, which means that it holds for all $s\in S$. Using this result in an induction argument and letting $k\to\infty$ we obtain the statement of the Lemma.
    \end{proof}
    

Intuitively, in the modified UMDP, the adversary is always allowed to pick a higher probability of transitioning to the unsafe state $s_{|S|}$, even when the last constraint in \eqref{eq:Gamma} is omitted. As a result, RDP returns a lower bound on the probability \eqref{eq:max_reachability_probability}. During the reminder of this section we consider that $\M$ has the structure described in Lemma~\ref{lemma:simplify}.

Then, on the modified $\M$, Alg.~\ref{alg:2O} computes the optimal choice of the adversary, for each state-action pair $(s,a)$ and iteration $k$, by extending the O-maximizing algorithm devised for IMDP value iteration \cite{givan2000bounded} to $\M$:  
via a 2-layer O-maximizing logic,
the algorithm efficiently allocates probability mass to states of $\M$ with the lowest value function while respecting the constraints in \eqref{eq:Gamma}. 
It begins by ensuring that the lower bounds $\underline P(s,a)(s')$ are satisfied for all states $s'$ (Lines 2-3), then proceeds to allocate mass to each cluster $q'\in Q(s,a)$ to meet the required lower bounds (Lines 4-12). The algorithm ensures that the total probability mass remains feasible by maintaining the constraints $\gamma(s') \le \overline P(s,a)(s')$, $\sum_{s'\in q'}\gamma(s') \le \overline P(s,a)(q')$ throughout the allocation (Lines 13-20). This allocation process guarantees that as much mass as possible is assigned to states with the smallest value function while ensuring $\gamma \in\Gamma_{s,a}$. Alg.~\ref{alg:2O} terminates once all mass is allocated. 
Note that RDP algorithm \cite{gracia2025efficient} calls Alg.~\ref{alg:2O} for every $(s,a)$ and in every iteration until termination.  The following theorem proves its correctness and runtime complexity.

\begin{theorem}[Correctness of Algorithm~\ref{alg:2O}]
\label{thm:soundness_alg}
    Let $k\in\naturals\cup\{0\}$, $s\in S\setminus\{s_{|S|}\}$, $a\in A$ and $\underline p^k\in\reals^{|S|}$ be as defined in Proposition~\ref{prop:robust_value_iteration}. Define $\text{Post}(s,a) := \{s' \in S : \overline P(s,a)(s') > 0\}$. Then, the output $\gamma$ of Algorithm~\ref{alg:2O} satisfies $\gamma \in \arg\min_{\gamma\in\Gamma_{s,a}}\sum_{s'\in S}\gamma(s') \underline p^k(s')$, and it has a computational complexity of $\mathcal{O}(|\text{Post}(s,a)|\log(|\text{Post}(s,a)|))$.
\end{theorem}
%
Note that the computational complexity of solving the linear program in the theorem statement using a standard Simplex algorithm is of $\mathcal{O}(|\text{Post}(s,a)|^3)$, highlighting the computational advantage of using Alg.~\ref{alg:2O}.
\begin{proof}
    We equivalently prove that any $\gamma$ generated by the algorithm belongs to $\Gamma_{s,a}$, and that $\gamma$ assigns the most probability to states $s'\in \text{Post}(s,a)$ with the smallest $\underline p^k(s')$. First we prove that $\gamma\in\Gamma_{s,a}$. From lines $2$ and $3$, we know that $\gamma(s') \ge \underline P(s,a)(s')$ for all $s'\in S$. Furthermore, by construction of $\M$ (see how Proposition~\ref{prop:hoeffding} computes the transition probability bounds and how these are modified in Lemma~\ref{lemma:simplify}), it is easy to verify that $\underline P(s,a)(q') \le \sum_{s'\in q'}\underline P(s,a)(s')$ for all, $q'\in Q(s,a)$, and therefore lines $4-14$ guarantee that each $q'\in Q(s,a)$ receives exactly $\underline P(s,a)(q')$ probability mass. Since, also by construction of $\M$, 
    it holds that $\underline P(s,a)(q') \le \sum_{s'\in q'}\overline P(s,a)(s')$ for all $q'$, this allocation of probability mass is always feasible, and line $13$ is reached with $m=0$. Additionally, by the logic of lines $10-11$, $\gamma(s') \le \overline P(s,a)(s')$ for all $s'$ so far. 
     In lines $13-20$ we allocate the remaining mass from $M$ to the states $s'\in \text{Post}(s,a)$ while respecting $\gamma(s') \le \overline P(s,a)(s')$ and $\gamma(q') \le \overline P(s,a)(q')$. Since $\Gamma_{s,a}\neq\emptyset$, it follows that $\sum_{s'\in \text{Post}(s,a)} \overline P(s,a)(s') \ge 1$
     , hence line $20$ is reached with $M=0$. Furthermore, by construction of $\M$ (see Proposition~\ref{prop:hoeffding} and Lemma~\ref{lemma:simplify}), the constraint $\gamma(q') \le \overline P(s,a)(q')$ never prevents all mass from $M$ to be allocated. Therefore, $\gamma\in\Gamma_{s,a}$. 
     Furthermore, since 1) the algorithm assigns first the least amount of mass that guarantees satisfaction of the lower bounds $\underline P(s,a)(s')$ and $\underline P(s,a)(q')$ for all $s',q'$, 2) when mass is assigned to $\gamma(s')$ (lines $2,3,11,20$), the states with smaller $\underline p^k(s')$ are considered first, and 3) the mass allocated to such states is the maximum amount that the upper bounds $\overline P(s,a)(s')$ and $\overline P(s,a)(q')$ allow. It follows that $\gamma$ minimizes the expression in Theorem~\ref{thm:soundness_alg}. Finally we prove the statement regarding the computational complexity. Note that the for loops in lines $6$ and $13$ are not nested. Additionally, since the sets $q'\in Q(s,a)$ are disjoint, the for loops in lines $6$ and $9$ are equivalent to a single for loop on $s'\in S$. Therefore lines $2$ through $20$ have complexity $\mathcal{O}(|\text{Post}(s,a)|)$, which is negligible in comparison with the complexity of sorting in line $1$, which is $\mathcal{O}(|\text{Post}(s,a)|\log(|\text{Post}(s,a)|))$. Therefore the algorithm has complexity $\mathcal{O}(|\text{Post}(s,a)|\log(|\text{Post}(s,a)|))$, which concludes the proof.
\end{proof}

\section{Case Studies}
\label{sec:case_studies}

We now demonstrate empirically the effectiveness of our approach through $7$ case studies. These include a nonlinear pendulum with non-additive 
disturbances, kinematic unicycle models with $2$- and $3$-D state-spaces and under nonlinear coulomb friction, a $2$-D linear system with multiplicative noise, and a $4$-D thermal regulation benchmark with multiplicative uncertainty. The considered specifications are reach-avoid ($\varphi_1$), the \LTLf specification from \cite{gracia2024data} ($\varphi_2$) and a $15$-step safety specification ($\varphi_3$). For details of these 
setups, see Appendix~\ref{app:appendix_B_arXiv}.

\begin{table}[t]
    \centering
    \caption{\small
    Benchmark results. ``Approach" indicates the abstraction used: UMDP (Definition~\ref{def:robust_mdp}), ``IMDP (Learn Support)" (a UMDP with only the first and third constraints in \eqref{eq:Gamma}), and ``Na\"ive IMDP" (traditional IMDP). $e_{avg}$ denotes the average difference in satisfaction probabilities between the lower and upper bounds across all states. Time for abstraction and synthesis is given in minutes, and $ N_\text{cluster}$ denotes the number of noise samples after clustering. The confidence in all case studies is $1-\alpha = 1-10^{-9}$ and $\beta_c = \alpha/2$.
    }
    \scalebox{0.66}{
    \begin{tabular}{c|c|c|c|c|c|c|c|c}
        System (Spec.) & Approach & $|Q|$ & $|A|$ & $N$ & $N_\text{cluster}$ & $e_{avg}$ & Abstr. Time & Synth. Time \\
        \midrule
        Pendulum ($\varphi_1$) & UMDP & $10^4$ & $5$ & $5\times 10^3$ & $47$ & $0.552$ & $1.796$ & $3.643$ \\
        \cite{gracia2024data} & UMDP & $10^4$ & $5$ & $10^4$ & $44$ & $0.007$ & $1.857$ & $6.679$ \\
        & UMDP & $10^4$ & $5$ & $10^5$ & $47$ & $0.003$ & $1.797$ & $2.515$ \\
        & \cite{gracia2024data} & $4\times10^4$ & $5$ & $10^{6}$ & $49$ & $0$ & $5.273$  & $61.167$ \\
        \midrule
        Pendulum ($\varphi_1$) & UMDP & $1.225\times 10^5$ & $5$ & $5\times10^4$ & $172$ & $0.108$ & $130.680$ & $203.817$ \\
        (Torque-Limited) & UMDP & $4\times 10^4$ & $5$ & $5\times10^4$ & $173$ & $0.440$ & $31.328$ & $15.203$ \\
        & UMDP & $4\times 10^4$ & $5$ & $10^5$ & $179$ & $0.186$ & $30.952$ & $40.241$ \\
        & UMDP & $4\times 10^4$ & $5$ & $10^6$ & $179$ & $0.048$ & $30.626$ & $29.936$ \\
        & UMDP & $4\times 10^4$ & $5$ & $10^7$ & $191$ &  $0.033$ & $40.520$ & $25.308$ \\
        & IMDP (Learn Support) & $4\times 10^4$ & $5$ & $10^7$ & $191$ &  $0.164$ & $16.714$ & $29.119$ \\
        \midrule
        $3D$ Unicycle ($\varphi_1$) & UMDP & $5.932\times10^4$ & $10$ & $10^{4}$ & $235$ & $0.579$ & $34.689$ & $33.156$ \\
        \cite{gracia2024data} & UMDP & $5.932\times10^4$ & $10$ & $10^{5}$ & $325$ & $0.267$ & $45.610$ & $36.18$ \\
        & UMDP & $5.932\times10^4$ & $10$ & $10^{6}$ & $358$ & $0.156$ & $52.733$ & $33.751$ \\
        & \cite{gracia2024data} & $6.4\times10^4$ & $10$ & $5\times10^{8}$ & $8869$ & $0.447$ & $457.431$ & $43.342$\\
        \midrule
        $3D$ Unicycle ($\varphi_1$) & UMDP & $7.401\times10^4$ & $10$ & $10^5$ & $248$ & $0.5$ & $68.708$ & $302.855$ \\
        (difficult) & UMDP & $5.932\times10^4$ & $10$ & $10^6$ & $285$ & $0.241$ & $61.583$ & $200.311$ \\
        & UMDP & $7.401\times10^4$ & $10$ & $10^6$ & $285$ & $0.165$ & $75.430$ & $207.516$ \\
        & IMDP (Learn Support) & $7.401\times10^4$ & $10$ & $10^6$ & $285$ & $0.166$ & $25.109$ & $21.606$ \\
        & Na\"ive IMDP & $5.932\times10^4$ & $10$ & $10^7$ & $340$ & $0.870$ & $45.865$ & $3.68$ \\
        & UMDP & $5.932\times10^4$ & $10$ & $10^7$ & $353$ & $0.170$ & $70.593$ & $178.599$ \\
        \midrule
        Multiplicative & UMDP & $3.6\times10^{3}$ & $1$ & $4.705\times10^{3}$  & $247$ & $0.369$ & $0.156$ & $0.051$ \\
        noise ($\varphi_1$) & IMDP (Learn Support) & $3.6\times10^{3}$ & $1$ & $4.371\times10^{3}$  & $247$ & $0.377$ & $0.156$ & $0.011$ \\
        \cite{skovbekk2023formal} & UMDP & $3.6\times10^{3}$ & $1$ & $4.725\times10^{4}$  & $306$ & $0.296$ & $0.194$ & $0.064$ \\
        & UMDP & $3.6\times10^{3}$ & $1$ & $4.722\times10^{5}$  & $312$ & $0.258$ & $0.199$ & $0.070$ \\
        & UMDP & $9.22\times10^{3}$ & $1$ & $4.722\times10^{5}$  & $309$ & $0.243$ & $0.594$ & $0.306$ \\
        & IMDP & $9.22\times10^{3}$ & $1$ & $4.722\times10^{5}$  & $309$ & $0.256$ & $0.594$ & $0.063$ \\
        &\cite{gracia2024data} & $10^{4}$ & $1$ & $4.66\times10^{5}$  & $1066$ & $0.323$ & $13.149$ & $3.863$ \\
        \midrule
        $2D$ Unicycle ($\varphi_2$) & IMDP (Learn Support) & $3.6\times10^3$ & $8$& $5\times10^{3}$ & $36$ & $0.135$  & $0.257$ & $4.527$\\
        \cite{gracia2024data} & UMDP & $3.6\times10^3$ & $8$& $5\times10^{3}$ & $36$ & $0.087$  & $0.257$ & $9.012$\\
         & UMDP & $3.6\times10^3$ & $8$& $10^{4}$ & $35$ & $0.038$  & $0.254$ & $7.013$\\
         & UMDP & $3.6\times10^3$ & $8$& $10^{5}$ & $41$ & $0.003$  & $0.294$ & $5.622$\\
         & UMDP & $3.6\times10^3$ & $8$& $10^{7}$ & $45$ & $0.000$  & $0.321$ & $4.341$\\
         & \cite{gracia2024data} & $3.6\times10^3$ & $8$ & $10^{7}$ & $46$ & $0.030$  & $0.106$ & $3.043$\\
         \midrule
         4-Room Heating ($\varphi_3$) & UMDP & $2.074\times 10^4$ & $16$ & $5\times 10^4$ & $620$ & $0.042$ & $143.869$ & $5.739$\\
         \cite{abate2010approximate} & IMDP (Learn Support) & $2.074\times 10^4$ & $16$ & $5\times 10^4$ & $620$ & $0.102$ & $129.381$ & $2.99$\\
         & UMDP & $2.074\times 10^4$ & $16$ & $5\times 10^5$ & $753$ & $0.020$ & $258.441$ & $6.127$\\
    \end{tabular}
    }
    \label{tab:case_studies}
\end{table}

We compare our approach against \cite{gracia2024data}, the only related work addressing the same problem, in Case Studies $1,3,5$-$6$. $2$ is a more challenging version of $1$, where $\w$ is unbounded and the pendulum cannot swing up in one go due to control saturation. Similarly, $4$ extends $3$ with unbounded noise of larger variance and a smaller goal set. Note that \cite{gracia2024data} relies on ambiguity set learning and cannot handle an unbounded $\w$.  We also show results obtained using a na\"ive IMDP abstraction with and without learned support, to show tightness of our approach. Additionally, we show the results of our approach for case studies $\#2$ and $\#6$ in Figure \ref{fig:results}, and for case study $\#4$ in Figure~\ref{fig:IMDP_example}.
 Table~\ref{tab:case_studies} summarizes our results, highlighting the clear advantages of our approach over \cite{gracia2024data}. Our method significantly reduces sample complexity, often by orders of magnitude, allowing for smaller abstractions while achieving tighter results, thus reducing computation times in most cases. For abstractions of the same size, our synthesis time is typically smaller, sometimes by an order of magnitude. Additionally, the table demonstrates that our approach produces tighter results (less error hence higher probabilistic guarantees) than using a na\"ive IMDP, showcasing the benefits of incorporating additional information into the definition of $\Gamma$, albeit with higher computational effort. Furthermore, UMDPs yield a smaller error, less than a half in some case studies, than when the abstraction is just an IMDP with an additional constraint related to learning the support of $P_W$, showing the benefit of a $2$-layer partition. We also compared the performance of Alg.~\ref{alg:2O} against the linear programming solver \emph{Linprog} on an abstraction with $|S| = 1600$ and $|A| = 8$, achieving the same guarantees but reducing the total synthesis time from 
$2880$s to $60$s, a reduction of $48\times$.

\begin{figure}[h]
    \begin{minipage}[b]{.3\linewidth}
        \centering
        \includegraphics[width=\linewidth]{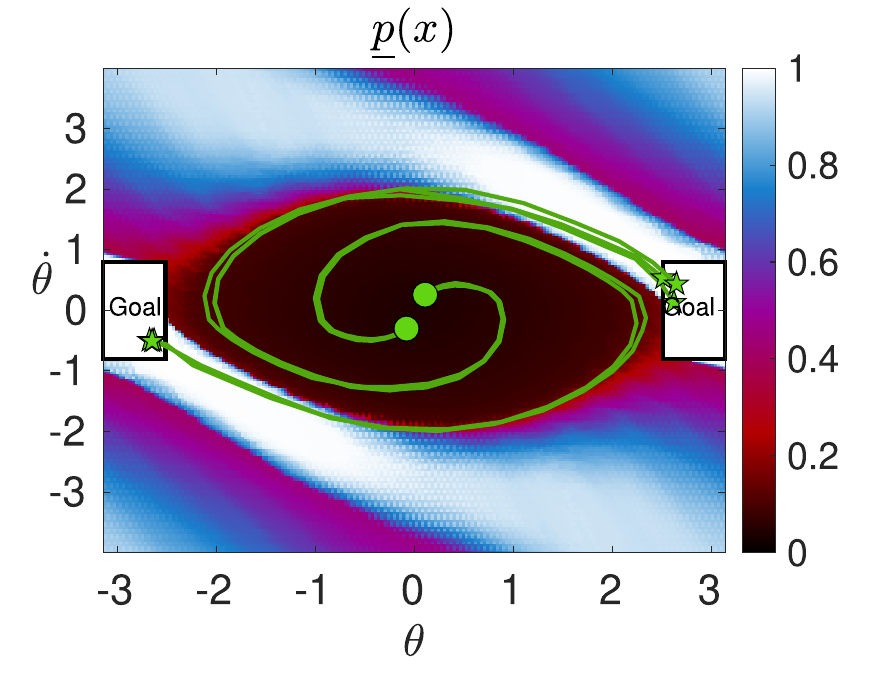}
        \subcaption{Pendulum ($N=5\times10^4$)}
    \end{minipage}
    \hfill
    \begin{minipage}[b]{.3\linewidth}
    \centering
        \includegraphics[width=\linewidth]{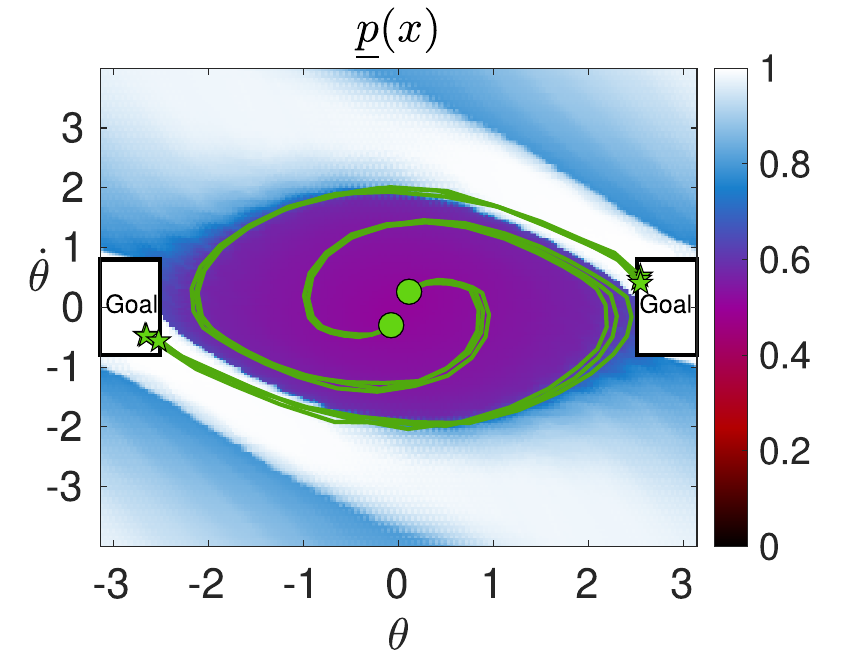}
        \subcaption{Pendulum ($N=10^5$)}
    \end{minipage}
    \hfill
    \begin{minipage}[b]{.3\linewidth}
    \centering
        \includegraphics[width=\linewidth]{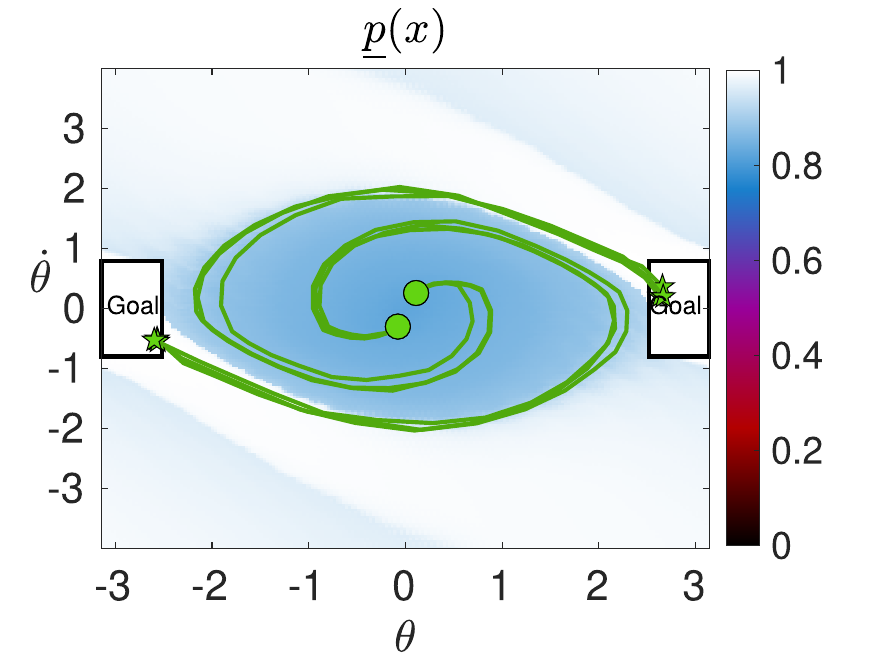}
        \subcaption{Pendulum ($N=10^6$)}
    \end{minipage}
    \hfill
    \begin{minipage}[b]{.3\linewidth}
        \centering
        \includegraphics[width=\linewidth]{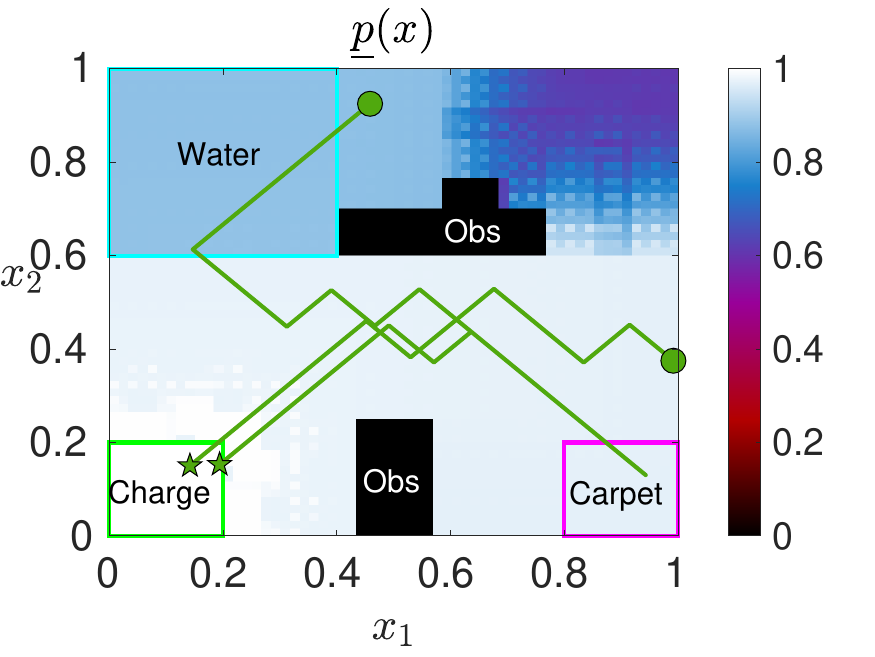}
        \subcaption{$2$D Unicycle ($N=5\times10^3$)}
    \end{minipage}
    \hfill
    \begin{minipage}[b]{.3\linewidth}
    \centering
        \includegraphics[width=\linewidth]{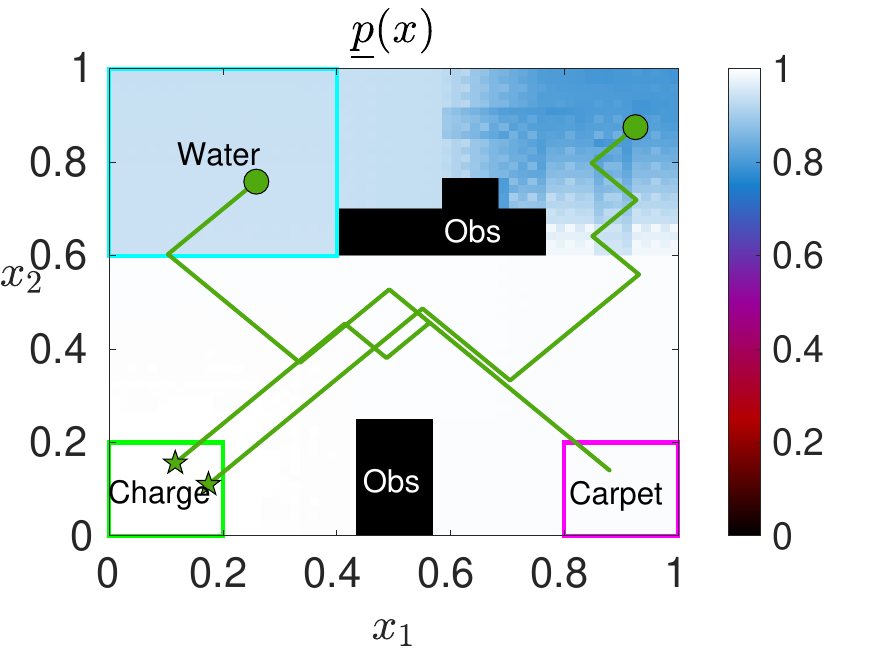}
        \subcaption{$2$D Unicycle ($N=10^4$)}
    \end{minipage}
    \hfill
    \begin{minipage}[b]{.3\linewidth}
    \centering
        \includegraphics[width=\linewidth]{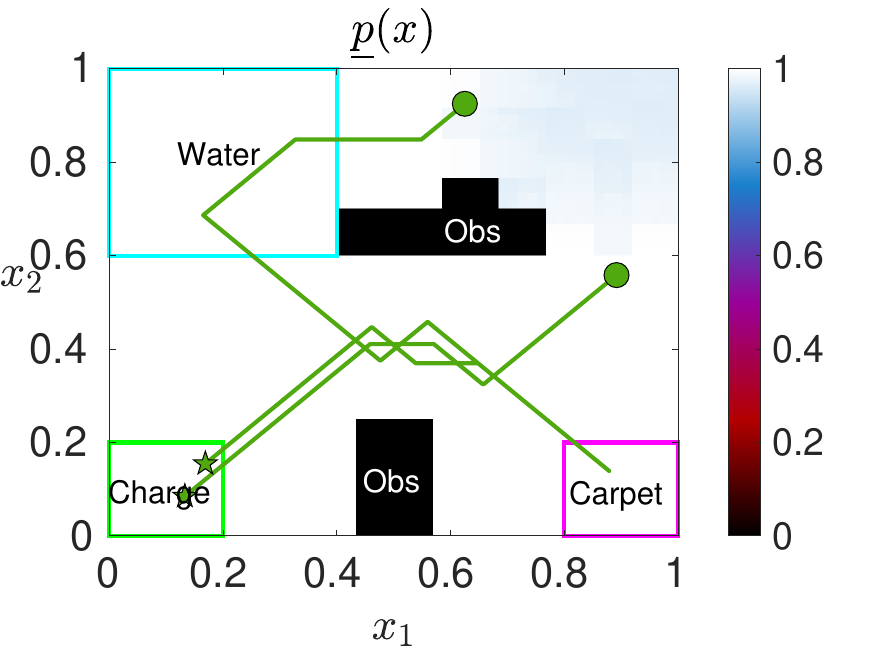}
        \subcaption{$2$D Unicycle ($N=10^5$)}
    \end{minipage}
    \caption{Synthesis results obtained in case studies $\#3$ (Figures (a-c)) and $\#6$ (Figures (d-f)). The figures show the lower bound in the probability of satisfying $\varphi_1$ and $\varphi_2$, respectively, along with trajectories of the system in closed loop with the synthesized strategy. It can be observed that increasing the number of samples raises this lower bound, and that all trajectories satisfy the respective specification.}
    \label{fig:results}
\end{figure}

\section{Conclusion}

We propose an approach to synthesize strategies for nonlinear stochastic systems with unknown disturbances via abstractions to UMDPs. We also identify pitfalls in the use of na\"ive abstractions for nonlinear systems and present a synthesis algorithm tailored to our UMDP class. Our extensive case studies show the efficacy and advantages of our framework w.r.t. existing works. In future research we plan to increase the tightness of our results by including additional information into the UMDP and investigate overlapping clusters.

\section{Acknowledgements}
I. Gracia and M. Lahijanian are supported in part by National Science Foundation (NSF) under grant number 2039062 and Air Force Research Lab (AFRL) under agreement number FA9453-22-2-0050.
L. Laurenti is partially supported by the NWO(grant OCENW.M.22.056).

\newpage

\appendix

\section{Measurability of the events in \eqref{eq:tp_bounds}}
\label{sec:measurability}

Let $r\in R$, $\tilde r\in\mathcal{B}(\reals^n)$ and $a\in A$, and consider the set-valued function $w\mapsto g(w) :=\text{Reach}(r,a,w)$, for all $w\in W$. Denote also by $B(0,\epsilon) \subset \reals^n$ the $\epsilon$-ball centered at the origin. We begin by stating the following technical lemma:
\begin{lemma}
\label{eq:lemma_epsilon_reach_set}
    For all $w\in W$, $\epsilon >0$, there exists a $\delta >0$ such that for all $w'\in W$ with $\|w - w'\| < \delta$, it holds that $g(w') \subseteq g(w)\oplus B(0,\epsilon)$.
\end{lemma}
\begin{proof}
Pick $w\in W$, $\epsilon >0$, and let $\delta = \epsilon/L_u$. We want to show that for all $w'\in W$ with $\|w - w'\| < \delta$ and $y'\in g(w')$ there exists a $y \in g(w)$ that is only $\epsilon$-apart. Then the statement in the proposition follows. Define $x$ such that $f(x,a,w') = y'$, which is possible by definition of $g(w')$, and let $y = f(x,a,w)$, which implies $y\in g(w)$. By Assumption~\ref{ass:lipschitz}, we obtain that $\|y-y'\| = \|f(x,a,w) - f(x,a,w')\| \le L_u\|w-w'\| < \epsilon$. Thus the proof is concluded.
\end{proof}
\begin{proposition}
    The events in Expression~\eqref{eq:tp_bounds} are measurable.
\end{proposition}
\begin{proof}
    Denote by $\mathcal{F}$, $\mathcal{G}$ and $\mathcal{K}$ respectively the sets of all closed, open and compact subsets of $\mathbb{R}^n$ with respect to the usual topology. Consider the Fell topology $\mathcal{T}(\mathcal{F})$ on $\mathcal{F}$, i.e. the one generated by the sets
\begin{align*}
    \mathcal{F}^K := \big\{F \in \mathcal{F}(\mathbb{R}^n) : F \bigcap K = \emptyset\big\}\qquad\text{and}\qquad
    \mathcal{F}_G := \big\{F \in \mathcal{F}(\mathbb{R}^n) : F \bigcap G \neq \emptyset\big\},
\end{align*}
for all $G\in\mathcal{G}$ and $K\in\mathcal{K}$. Denote by $\mathcal{B(\mathcal{F})}$ the Borel $\sigma$-algebra on $(\mathcal{F}, \mathcal{T}(\mathcal{F}))$. Denote also by $\mathcal{T}(W)$ the usual topology on $W$ and by $\mathcal{B}(W)$ the Borel $\sigma$-algebra on $W$. For simplicity, assume that $\text{Reach}(r,a,w)$ is closed and note that if it is not, we can always consider its closure without compromising the soundness of the approach. First, we establish that the function $g$ is $(W, \mathcal{B}(W))$-$(\mathcal{F},\mathcal{B}(\mathcal{F}))$ measurable by proving a sufficient condition, namely, upper semi-continuity of $g$: $g$ is upper semi-continuous if $g^-(K) := \{w\in W : g(w)\bigcap K \neq \emptyset\}$ is closed on $(W,\mathcal{T}(W))$ for all $K\in\mathcal{K}$. We proceed by contradiction: assume that $g^-(K)$ is not closed for some $K\in\mathcal{K}$. Then, there must exist a $w_\infty\in W\setminus g^-(K)$ and a sequence $\{w_n\}_{n\in\naturals} \subset g^-(K)$ that converges to $w_\infty$. Since $K$ and $g(w_\infty)$ are closed and non-intersecting, there exists $\epsilon>0$ such that $(g(w_\infty)\oplus B(0,\epsilon))\bigcap K = \emptyset$. 
By convergence of $\{w_n\}_{n\in\mathbb{N}}$, for every $\delta >0$, there exists $N \in \mathbb{N}$ such that $\|w_\infty-w_N\| < \delta$. This implies, by Lemma~\ref{eq:lemma_epsilon_reach_set}, that we can choose $N$ large enough so that we ensure, $g(w_N) \subseteq g(w_\infty)\oplus B(0,\epsilon)$. It follows that $g(w_N)\bigcap K \subseteq (g(w_\infty)\oplus B(0,\epsilon))\bigcap K = \emptyset$, which proves that $w_N\notin g^-(K)$, being this result a contradiction. Therefore $g^-(K)$ is closed in $(W,\mathcal{T}(W))$ for all $K\in\mathcal{K}$, making $g$ upper-semicontinuous and measurable, i.e., $w \mapsto \text{Reach}(r,a,w)$ is a well-defined random set \cite{molchanov2005theory}. Finally, since $\tilde r\in\mathcal{B}(\reals^n)$, 
the event $\{w\in W : g(w) \bigcap \tilde r \neq \emptyset\}$ is measurable, i.e., it belongs to $\mathcal{B}(W)$ \cite[Theorem 2.3]{molchanov2005theory}. Furthermore, since $\{w\in W : g(w) \subseteq\tilde r\} = W\setminus\{w\in W : g(w) \bigcap \reals^n\setminus\tilde r \neq \emptyset\}$ and $\reals^n\setminus\tilde r\in\mathcal{B}(\reals^n)$, the event $\{w\in W : g(w) \subseteq \tilde r\}$ is also in $\mathcal{B}(W)$, thus concluding the proof.
\end{proof}
\section{Details of the Case Studies}
\label{app:appendix_B_arXiv}

In this section we describe in detail our case-studies, and provide precise values of all system and approach-related parameters.

We consider $3$ different specifications:
\begin{itemize}
    \item Reach-avoid $\varphi_1 := \Box (\neg \Prop_\text{unsafe}) \land \Diamond \text{goal}$,
    \item Complex \LTLf specification $\varphi_2 := \Box(\neg\Prop_\text{unsafe}) \land \Box(\rm{water}\rightarrow(\neg \rm{charge}\: \mathcal{U}\rm{carpet}))\land\Diamond(\rm{charge})$ considered in \cite{vazquez2018learning} and representing the task of reaching a charge station while remaining safe and, if the system goes through a region with water, then first drying in a carpet before charging,
    \item $15$-step safety specification $\varphi_3 := \Box^{\le 15} (\neg \Prop_\text{unsafe})$ considered in \cite{abate2010approximate},
\end{itemize}
where $\Box^{\le 15}$ is the \emph{bounded globally} operator, interpreted as ``something happens at all time steps less or equal than $15$''. The confidence level in all case studies is $1-\alpha = 1-10^{-9}$ and $\beta_c = \alpha/2$.

\subsection{Pendulum}

Case studies $\#1$ and $\#2$ both consider the problem of swinging up a pendulum, starting from the downward orientation, without ever exceeding some limits in its angular velocity. The state of the system is $2$-dimensional, composed by the angle $\boldsymbol{\theta}_k$ w.r.t. the vertical (downward) position, and the angular velocity $\boldsymbol{\dot\theta}_k$, whereas the control input $u_k$ is the torque applied at the joint. The disturbance $\w_k$ corresponds to a horizontal wind disturbance, which generates an aerodynamic drag force that is quadratic in $\w_k$ and depends on the state in a nonlinear fashion. The system dynamics are
\begin{align*}
    \begin{bmatrix}
        \boldsymbol{\theta}_{k+1}\\
        \boldsymbol{\dot\theta}_{k+1}
    \end{bmatrix} = \begin{bmatrix}
        \boldsymbol{\theta}_k + \Delta t \boldsymbol{\dot\theta}_k\\
        \boldsymbol{\dot\theta}_k + \Delta t \big( -c_d\:\text{sign}(l\boldsymbol{\dot\theta}_k - \w_k\cos(\boldsymbol{\theta}_k)) (l\boldsymbol{\dot\theta}_k - \w_k\cos(\boldsymbol{\theta}_k))^2 - \sin(\boldsymbol{\theta}_k) + u_k \big)
    \end{bmatrix}.
\end{align*}
In the first case study in Table~\ref{tab:case_studies}, we let $\Delta t = 0.25$, $c_d = 0.3$ and $\w_k$ be distributed according to a zero-mean Gaussian distribution with covariance $0.04$, truncated in the interval $[-1,1]$. Note that even though our approach can handle unbounded disturbances, for this system to satisfy Assumption~\ref{ass:lipschitz}, we need the support $W$ of $P_W$ to be bounded, even if the exact bounds are unknown to our approach. We also let $X$ be the set of all states $[\theta,\dot\theta]$ such that $|\dot\theta| \le 3$, and define $U$ by uniformly discretizing the interval $[-0.8, 0.8]$ into $5$ values. Furthermore, we define the $\text{goal}$ set as the set of states $[\theta, \dot\theta]$ such that $\theta$ is within $0.628$ of the upward position and $|\dot\theta| \le 0.6$. On the other hand, in the second case study in Table~\ref{tab:case_studies} we let $\Delta t = 0.3$, $c_d = 0.2$ and $\w_k \sim \mathcal{N}(0, 0.049)$ truncated on the interval $[-3, 3]$. Again, in our approach we let the support $W$ of $P_W$ be unknown. We also let $X$ contain all states with an angular velocity $|\dot\theta| \le 4$, and define $U$ by discretizing the interval $[-0.415, 0.415]$ as explained before. Note that the limit in the maximum control input in case study $\#2$ makes it impossible to swing up the pendulum in one attempt, hence making it necessary to swing the pendulum back and forth repeatedly to achieve the swing up. For this reason, we refer to case study $\#2$ as ``torque limited''. Furthermore, we let the $\text{goal}$ set contain all states $[\theta, \dot\theta]$ such that $\theta$ is within $0.628$ of the upward position and $|\dot\theta| \le 0.6$ (see Figure~\ref{fig:results}(a-c)). In both case studies we partition the state-space via a uniform grid, and learn the support of $P_W$ by setting $\epsilon_c = 0.001$. When constructing our UMDP abstractions, we define the coarse layer $Q$ by uniformly clustering the regions in $R$. Each cluster $q$ contains the abstract states in $S$ corresponding to $2\times 2$ regions in $R$, as shown in Figure~\ref{fig:figure}.

\subsection{3D Unicycle}

The system is a kinematic model of a unicycle, where the state is $3$-dimensional, composed by the $2$D components $[\boldsymbol{x}_k, \boldsymbol{y}_k, \boldsymbol{\theta}_k]$ of the position in the plane and the unicycle's heading angle. The control inputs are the linear velocity and the yaw (heading angle) rate. The disturbance $\w_k$ represents \emph{Coulomb friction}, which generates a force in the opposite direction of the linear velocity. Its dynamics are the following:
\begin{align}
\label{eq:3D_unicycle}
    \begin{bmatrix}
        \boldsymbol{x}_{k+1}\\
        \boldsymbol{y}_{k+1}\\
        \boldsymbol{\theta}_{k+1}
    \end{bmatrix} = \begin{bmatrix}
        \boldsymbol{x}_{k} + \Delta t ( u_k^{(1)} - c_d^{(1)} \w_k^{(1)} )\cos(
        \boldsymbol{\theta}_{k})\\
        \boldsymbol{y}_{k} + \Delta t ( u_k^{(1)} - c_d^{(1)} \w_k^{(1)} )\sin(
        \boldsymbol{\theta}_{k})\\
        \boldsymbol{\theta}_k + \Delta t\: u_k^{(2)} + c_d^{(2)} \w_k^{(2)}
    \end{bmatrix}
\end{align}
Note that the effect of the disturbance is nonlinear in the state. In the third and fourth case studies in Table~\ref{tab:case_studies}, we let $\Delta t = 0.5$ and $c_d = [0.1, 0.05]$. We also let $X$ be the set of all states whose position is in $[0,1]^2$, except for a rectangular obstacle in the center of $X$ whose edges are $0.154$ long (see Figure~\ref{fig:IMDP_example}). In case study $\#3$, $\w_k$ is distributed according to a Gaussian distribution with mean $[0.4, 0]^T$ and covariance $\text{diag}(0.067^2, 0.067^2)$, truncated on a ball of radius $1$ centered on its mean. In this case study the support $W$ of $P_W$ is known. For this reason, we define $\widehat{\text{Post}}(s,a)$, and therefore our UMDP abstraction, using $W$ instead of $\widehat W$, and let $\epsilon_c = 0$, as the successor state of the pair $(s,a)$ will be in $\widehat{\text{Post}}(s,a)$ with probability $1$. We also define $U := \{0.21, 0.3\}\times\{ -2, -1, 0, 1, 2 \}$. Furthermore, the $\text{goal}$ set is defined as the set of states whose position components lie inside the box $[0, 0.359]^2$. On the other hand, in case study $\#4$, $\w_k$ is distributed according to a non-truncated Gaussian distribution with the same mean as in case study $\#3$, but a significantly larger covariance of $\text{diag}(0.2^2, 0.2^2)$. We also let $U := \{0.15, 0.3\}\times\{ -2, -1, 0, 1, 2 \}$, and define the $\text{goal}$ set as the smaller box $[0, 0.256]^2$. Given the smaller goal set and the higher covariance and support of the disturbance, we denote case study $\#4$ as ``$3$D unicycle (difficult)'', and we learn the support of $P_W$ by setting $\epsilon_c = 0.001$. In both case studies we partition the state-space via a uniform grid. When constructing our UMDP abstractions, we define the coarse layer $Q$ by uniformly clustering the regions in $R$. Each cluster $q$ contains the abstract states in $S$ corresponding to $3\times 3\times 3$ regions in $R$ (see Figure~\ref{fig:figure} for the intuition).

\subsection{Multiplicative Noise}

Case study $\#5$ is a $2$-dimensional system with multiplicative noise taken from \cite{skovbekk2023formal}. In this case study, the support $W$ of $P_W$ is known. Because of this, we use it in Definition~\ref{def:umdp abstraction} instead of $\widehat W$. When constructing our UMDP abstractions, we define the coarse layer $Q$ by uniformly clustering the regions in $R$. Each cluster $q$ contains the abstract states in $S$ corresponding to $3\times 3$ regions in $R$ (see Figure~\ref{fig:figure} for the intuition).

\subsection{2D Unicycle}

Case study $\#6$ is a $2$-dimensional unicycle model obtained from  \eqref{eq:3D_unicycle} by fixing the first component of the input to $0.3$ and by considering the heading angle $\theta_k$ as the input. The latter takes values in the set $U$, obtained via a uniform discretization of the interval $[-\pi,\pi]$ into $8$ values. We consider $c_d=0.2$, $\Delta t = 0.5$, and let $\w_k$ be distributed according to a Gaussian with mean $0.4$ and variance $0.067^2$, truncated on the interval $[-0.6, 1.4]$. We let this support be known, and therefore we define $\widehat{\text{Post}}(s,a)$, using $W$ instead of $\widehat W$, and let $\epsilon_c = 0$, as the successor state of the pair $(s,a)$ will be in $\widehat{\text{Post}}(s,a)$ with probability $1$. The size and position of the regions of interest is depicted in Figure~\ref{fig:results} (d-f). We discretize the state space uniformly and, when constructing our UMDP abstraction, we define the coarse layer $Q$ in such a way that each cluster $q$ corresponds to $2\times 2$ adjacent regions in the partition $R$, i.e., each cluster corresponds to $4$ abstract states in $S$, as shown in Figure~\ref{fig:figure}.

\subsection{4-Room Heating}

Case study $\#7$ is a $4$-dimensional system whose state is composed by the temperatures of $4$ rooms, and its dynamics are taken from \cite{abate2010approximate} and modified to make the noise multiplicative:
\begin{align*}
    \x_{k+1} = \text{diag}(1 + \w_k)A\x_t + b + b_u u_k,
\end{align*}
with
\begin{align*}
    A := \begin{bmatrix}
    0.901 &0.0625 &0 &0\\
     0.0625 & 0.839 & 0.0625 &0 &0\\
     0 & 0.0625 & 0.839 & 0.0625\\
     &0 &0 & 0.0625 &0.901
    \end{bmatrix}, \quad b := \begin{bmatrix}
        0.219\\
        0.219\\
        0.219\\
        0.219
    \end{bmatrix},
\end{align*}
and $b_u = \text{diag}(1, 1, 1, 1)$. We let each component of $\w_k$ be Gaussianly distributed with mean $0$ and covariance $4.44\times 10^{-5}$, and let all components be independent from each other. We define the safe set as $X := [18.5, 23.5]^4$. Note that even though the covariance of the disturbance is small, its effect on the dynamics is way bigger because of the multiplicative nature of $\w_k$. We define the set $U$ as the set of all binary-valued vectors in $\reals^n$, meaning that each control input $u\in U$ is such that $u^{(i)} = 1$ if the radiator of the $i$-th room is on and $0$ if it is off, for all $i\in \{1,2,3,4\}$. We partition the state-space uniformly to obtain the abstraction, and use $\epsilon_c = 0.001$. When constructing our UMDP abstractions, we define the coarse layer $Q$ by uniformly clustering the regions in $R$. Each cluster $q$ contains the abstract states in $S$ corresponding to $3\times 3 \times 3 \times 3$ regions in $R$ (see Figure~\ref{fig:figure} for the intuition).

\bibliographystyle{IEEEtran}\bibliography{yourbibfile}

\end{document}